\documentclass[11pt]{article}

\usepackage[labelfont=sc]{caption}

\usepackage[backend=bibtex,  firstinits=true, citestyle=authoryear, maxbibnames=4]{biblatex}
\usepackage[utf8]{inputenc}
\usepackage{graphicx}

\usepackage{url}

\usepackage{color}
\usepackage[dvipsnames]{xcolor}

\usepackage[a4paper, margin=2cm]{geometry}

\usepackage{algorithm,algpseudocode}

\definecolor{mygray}{gray}{0.6}

\usepackage[toc,page]{appendix}
\usepackage{authblk}

\usepackage{titlesec}
\titleformat*{\section}{\large\bfseries}
\titleformat*{\subsection}{\bfseries}
\titleformat*{\subsubsection}{\itshape}
\usepackage{bbm}

\usepackage{bm}

\definecolor{DarkMidnightBlue}{rgb}{0.0, 0.2, 0.5}
\usepackage[colorlinks=true,linkcolor=Maroon,citecolor=DarkMidnightBlue,urlcolor=black]{hyperref}

\usepackage{lineno}
\usepackage{booktabs}

\usepackage{amsmath}
\usepackage{newunicodechar}

\newcommand{\niton}{\not\owns}

\usepackage{amsthm,amsmath}
\usepackage{amssymb,verbatim}

\makeatletter
\def\lstAZ{A, B, C, D, E, F, G, H, I, J, K, L, M, N, O, P, Q, R, S, T, U, V, W, X, Y, Z}
\def\lstaz{a, b, c, d, e, f, g, h, i, j, k, l, m, n, o, p, q, r, s, t, u, v, w, x, y, z}
\def\lstAZBB{B, C, D, E, F, G, H, I, J, K, L, M, N, O, P, Q, R, T, U, V, W, X, Y, Z}

\newcommand{\MkScr}[1]{\expandafter\def\csname s#1\endcsname{\mathscr{#1}}}
\newcommand{\MkUp}[1]{\expandafter\def\csname u#1\endcsname{\mathrm{#1}}}
\newcommand{\MkFrak}[1]{\expandafter\def\csname f#1\endcsname{\mathfrak{#1}}}
\newcommand{\MkCal}[1]{\expandafter\def\csname c#1\endcsname{\mathcal{#1}}}
\newcommand{\MkBB}[1]{\expandafter\def\csname #1#1\endcsname{\mathbb{#1}}}

\@for\i:=\lstAZ\do{%
	\expandafter\MkScr \i  %
	\expandafter\MkFrak \i  %
	\expandafter\MkUp \i %
	\expandafter\MkCal \i  %
		  }    
\@for\i:=\lstaz\do{%
	\expandafter\MkUp \i   }    
	
\@for\i:=\lstAZBB\do{%
	\expandafter\MkBB \i     }
	    
\makeatother

\newcommand{\E}{\mathrm{E}}

\newcommand{\ind}{\mathbf{1}}

\newcommand{\dd}{ \mathrm{d}}

\newcommand{\Ge}{ \mathcal{L}}
\newcommand{\sign}{\operatorname{sign}}

\DeclareMathOperator*{\argmin}{argmin}
\renewcommand{\epsilon}{\varepsilon}

\renewcommand{\tilde}{\widetilde}
\numberwithin{equation}{section}

\newtheorem{theorem}{Theorem}[section]
\newtheorem{lemma}[theorem]{Lemma}
\newtheorem{proposition}[theorem]{Proposition}

\theoremstyle{definition}

\newtheorem{remark}[theorem]{Remark}

\newtheorem{assumption}[theorem]{Assumption}

\usepackage{algorithm}
\usepackage{algpseudocode}

\usepackage{stmaryrd}

\newcommand{\abovetozero}{\circ\!\shortleftarrow }
\newcommand{\belowtozero}{\shortrightarrow\!\circ}
\newcommand{\aboveaway}{\circ\!\shortrightarrow}
\newcommand{\belowaway}{\shortleftarrow\!\circ}
\newcommand{\zerofromabove}{\stackrel{\shortleftarrow}{\circ}}
\newcommand{\zerofrombelow}{\stackrel{\shortrightarrow}{\circ}}
\newcommand{\towardzero}{\shortrightarrow\!\circ\!\shortleftarrow}
\newcommand{\awayfromzero}{\shortleftarrow\!\circ\!\shortrightarrow}
\newcommand{\atzero}{\circ}

\newcommand{\R}{\mathbb R}
\newcommand{\1}{\mathbbm 1}

\usepackage{pgfplots}
\definecolor{lime}{HTML}{A6CE39}
\DeclareRobustCommand{\orcidicon}{%
	\begin{tikzpicture}
	\draw[lime, fill=lime] (0,0) 
	circle [radius=0.16] 
	node[white] {{\fontfamily{qag}\selectfont \tiny ID}};
	\draw[white, fill=white] (-0.0625,0.095) 
	circle [radius=0.007];
	\end{tikzpicture}
	\hspace{-2mm}
}

\foreach \x in {A, ..., Z}{%
	\expandafter\xdef\csname orcid\x\endcsname{\noexpand\href{https://orcid.org/\csname orcidauthor\x\endcsname}{\noexpand\orcidicon}}
}

\title{{\bf Sticky PDMP samplers for sparse and local inference problems}}
\author[1]{{\small \bf Joris Bierkens}\orcidA{}}
\author[2]{{\small \bf Sebastiano Grazzi\orcidB{}}}
\author[3]{\\{\small \bf Frank van der Meulen}\orcidC{}}
\author[4]{{\small\bf Moritz Schauer}\orcidD{}}
\affil[1]{\small Delft Institute of Applied Mathematics (DIAM), Delft University of Technology, The Netherlands}
\affil[2]{\small Department of Statistics, University of Warwick, United Kingdom}
\affil[3]{\small Department of Mathematics,  Vrije Universiteit Amsterdam, The Netherlands}
\affil[4]{\small Department of Mathematical Sciences,  Chalmers University of Technology, Sweden and University of Gothenburg, Sweden}

\date{\today}
\pgfplotsset{compat=1.17} 
\begin{document}
\maketitle
\begin{abstract}
\noindent We construct a new class of efficient Monte Carlo methods based on continuous-time piecewise deterministic Markov processes (PDMPs) suitable for inference in high dimensional sparse models, i.e.\  models for which there is prior knowledge that many coordinates are likely to be exactly $0$. This is achieved with the fairly simple idea of endowing existing PDMP samplers with ``sticky'' coordinate axes, coordinate planes etc. Upon hitting those subspaces, an event is triggered during which the process \emph{sticks} to the subspace, this way spending some time in a sub-model. This results in  \emph{non-reversible} jumps between different (sub-)models. While we show that PDMP samplers in general can be made sticky, we mainly focus on the Zig-Zag sampler.  
Compared to the Gibbs sampler for variable selection, we heuristically derive favourable dependence of the Sticky Zig-Zag sampler on dimension and data size. The computational efficiency of the Sticky Zig-Zag sampler is further established through numerical experiments where both the sample size and the dimension of the parameter space are large.
\end{abstract}


\section{Introduction}\label{sec: intro}
\subsection{Overview}
Consider the problem of simulating from a measure $\mu$ on $\RR^d$ that is a mixture of atomic and continuous components.
A key application  is Bayesian inference for sparse problems and variable selection  under a spike-and-slab prior $\mu_0$ of the form 
\begin{equation}\label{eq: spike-and-slab prior}
\mu_0(\dd x) = \prod_{i = 1}^d \left( w_i \pi_{i}(x_i)\dd x_i + (1-w_i) \delta_0(\dd x_i)\right). 
\end{equation}
Here, $w_i \in [0,1]$, $\pi_{1}, \pi_{2},\dots,\pi_d$ are densities with respect to the Lebesgue measure referred to as \emph{slabs} and $\delta_0$ denotes the Dirac measure at zero. 
For sampling from $\mu$, it is common  to construct  and simulate a Markov process with $\mu$ as invariant measure. Routinely used samplers such as the Hamiltonian Monte Carlo sampler (\cite{duane1987hybrid}) cannot be applied directly due to the degenerate nature of $\mu$. We show that ``ordinary'' samplers based on piecewise deterministic Markov processes (PDMPs) can be adapted to sample from $\mu$ by introducing \emph{stickiness}. 
     
In piecewise deterministic Markov processes, the state space is augmented by adding to each coordinate $x_i$ a velocity component $v_i$, doubling the dimension of the state space. They are characterized by piecewise deterministic dynamics between event times, where event times correspond to changes of velocities. 
PDMPs have received recent attention because they have good mixing properties (they are non-reversible and have `momentum', see e.g. \cite{andrieu2019peskuntierney}), they take gradient information into account and they are attractive in Bayesian inference scenarios with a large number of observations because they allow for subsampling of the observations without creating bias (\cite{bierkens2019}, \cite{bierkens2020boomerang}).  

We introduce ``sticking event times'', which occur every time a coordinate of the process state hits $0$. At such a time that particular component of the state freezes for an independent exponentially distributed time with a specifically chosen rate equal to $|v_i|\kappa_i$, for some $\kappa_i>0$ which depends on $\mu$. This  corresponds to temporarily setting the marginal velocity to $0$: the process ``sticks to (or freezes at) 0'' in that coordinate, while the other coordinates keep moving, as long as they are not stuck themselves. After the exponentially distributed time the  coordinate moves again with its original velocity, see  Figure~\ref{fig: Zig-Zag} for an illustration of the sticky version of the Zig-Zag sampler (\cite{bierkens2019}). By this we mean that the dynamics of a ordinary PDMP are adjusted such that  the process can spend a positive amount of time at the origin, at the coordinate axes and at the coordinate (hyper-)planes by sticking to $0$ in each coordinate for a random time span whenever the process hits $0$ in that particular coordinate.  By restoring the original velocity of each coordinate after sticking at 0, we  effectively generate \emph{non-reversible jumps between states with different sets of non-zero coordinates}. In the Bayesian context this corresponds to having non-reversible jumps between models of varying dimensionality.

This allows us to construct a piecewise deterministic process that has a pre-specified measure $\mu$ as invariant measure, which we  assume to be of the  form
\begin{equation}
    \label{eq: target measure}
    \mu(\dd x) = C_\mu \exp(-\Psi(x))\prod_{i=1}^d \left( \dd x_i + \frac1{\kappa_i} \delta_{0}(\dd x_i)\right) 
\end{equation}
for some differentiable  function $\Psi$, normalising constant $\, C_\mu> 0$ and positive parameters $\kappa_1,\kappa_2,\dots,\kappa_d$. Here the Dirac masses are located at $0$, but generalizations are straightforward. The resulting samplers and processes are referred to as \emph{sticky samplers} and \emph{sticky piecewise deterministic Markov processes} respectively. 
The proportionality constant $C_\mu$ is assumed to be unknown while $(\kappa_i)_{i=1,\dots,d}$ are known. This is a natural assumption; suppose a statistical model with parameter $x$ and log-likelihood $\ell(x)$ (notationally, we drop the dependence of $\ell$ on the data). Under the spike-and-slab prior defined in Equation~\eqref{eq: spike-and-slab prior}, the posterior measure is of the form of Equation~\eqref{eq: target measure} with
\begin{equation}
    \label{eq: equivalence spike-and-slab}
     \Psi(x) = C - \ell(x) - \sum_{i=1}^d  \log(\pi_{i}(x_i)), \quad \kappa_i =  \frac{w_i}{1-w_i}\pi_i(0) 
\end{equation}
where $C$, independent of $x$, can be chosen freely for convenience. A popular choice for $\pi_i$ is a Gaussian density centered at $0$ with standard deviation $\sigma_i$. In this case, as  $w/(1-w) \approx w$ for $w\approx 0$,   $\kappa_i$ depends linearly on $w_i/\sigma_i$ in the sparse setting.

\begin{figure}[ht!]
    \centering
    \includegraphics[width = 0.8\linewidth]{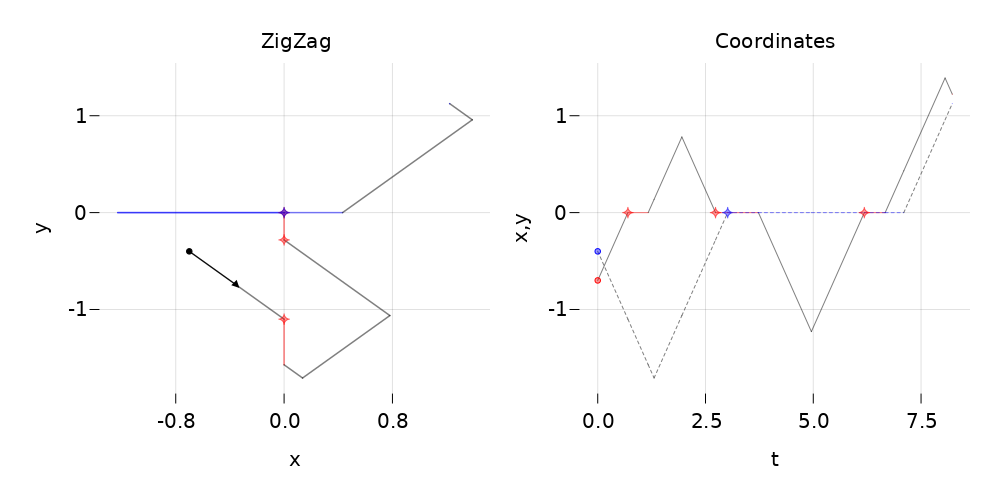}
    \caption{2-dimensional Sticky Zig-Zag sampler with initial position $(-0.75, -0.4)$ and initial velocity $(+1, -1)$. On the left panel, a trajectory on the $(x,y)$-plane of the Sticky Zig-Zag sampler. The sticky event times relative to the $x$ (respectively $y$) coordinate and the trajectories with the $x$ (respectively $y$) stuck at 0 are marked with a blue (respectively red) cross and line. On the right panel, the trajectories of each coordinate  against the time using the same (color-) scheme. The trajectory of $y$ is dashed.}
    \label{fig: Zig-Zag}
\end{figure}

Relevant quantities useful for model selection, such as the posterior probability of  a model excluding the first variable
{\small \[
    \mu(\{0\}\times \RR^{d-1}) = C_\mu \int \exp(-\Psi(x)) \frac{1}{\kappa_1} \delta_0(\dd x_1) \prod_{i=2}^d \left( \dd x_i + \frac1\kappa_i \delta_{0}(\dd x_i)\right) 
\]}%
cannot be directly computed if  $C_\mu$ is unknown. However, given a trajectory $\left(x(t)\right)_{0\le t\le T}$ of a PDMP  with invariant measure $\mu$, the quantity $\mu(\{0\}\times \RR^{d-1})$ can be approximated by  the ratio $T_0/T$ where $T_{0} =  \mbox{Leb}\{0\le t \le T\colon x_1(t) = 0\}$. This simple, yet general idea requires the user only to specify $\{\kappa_i\}_{i=1}^d$ and $\Psi$ as in Equation~\eqref{eq: target measure}. Moreover,  the posterior probability that a collection of variables are all jointly equal to zero can be estimated in a similar way by computing the fraction of time that all corresponding coordinates of the process are simultaneously zero and, more generally, expectations of functionals with respect to the posterior can be estimated from the simulated trajectory.

\subsection{Related literature}\label{sec: related literature}

The main purpose of this paper is to show how ``ordinary'' PDMPs can be adjusted to sample from the measure $\mu$ as defined in \eqref{eq: target measure}. The numerical examples illustrate its applicability in a wide range of applications. One specific application that has received much attention in the statistical literature is variable selection using a spike-and-slab prior. For the linear model, early contributions include \textcite{mitchell1988bayesian} and \textcite{george1993variable}. Some later contributions for hierarchical models derived from the linear model are \textcite{Ishwaran_2005}, \textcite{guan2011bayesian},  \textcite{zanella2018scalable} and \textcite{liang2021adaptive}. These works have in common that samples from the posterior are obtained from Gibbs sampling  and can be implemented in practise only in specific cases (when the Bayes factors between (sub-)models can be explicitly computed). A general and common framework for MCMC methods for variable selection was introduced in \textcite{Green95reversiblejump} and \textcite{green2009reversible} and referred to as \textit{reversible jump MCMC}.

Methods that  scale better  (compared to Gibbs sampling) with  either the  sample size or dimension of the parameter can be obtained in different ways. Firstly, rather than sampling from the posterior one can {\it approximate} the posterior within a specified class, for example using variational inference.  As an example, \textcite{ray2020spike} adopt this approach in a logistic regression problem with  spike-and-slab prior. Secondly, one can try to obtain sparsity  using a prior which is not of spike-and-slab type. For example, \textcite{griffin2021bayesian} consider   Gibbs sampling algorithms  for the linear model with  priors that are designed to promote sparseness, such as the Laplace or horseshoe prior (on the parameter vector). While such methods scale well with dimension of data and parameter, these target a different problem: the  posterior is not of the form \eqref{eq: target measure}. That is, the posterior itself is not sparse (though derived point estimates may be sparse and the posterior itself may have good properties when viewed from a frequentist perspective). Moreover, part of the computational efficiency is related to the specific model considered (linear  or logistic regression model) and, arguably,  a generic gradient-based MCMC method would perform poorly on such measures since the gradient of the \mbox{(log-)density} near 0 in each coordinate explodes to account for the change of mass in the neighborhood of 0 induced by the continuous spike component of the prior.

 A recent related work by  \textcite{chevallier2020reversible} addresses variable selection problems using PDMP samplers. The different approach taken in that paper is  based on  the framework of reversible jump (RJ) MCMC as proposed in  \textcite{Green95reversiblejump}. A comparison between \textcite{chevallier2020reversible} and our work may be found in Appendix~\ref{app: Comparison between reversible jump PDMPs and sticky PDMPs}. 

\subsection{Contributions}
\begin{itemize}
    \item We show how to construct sticky PDMP samplers from ordinary PDMP samplers for sampling from the measure in Equation~\eqref{eq: target measure}. This extension allows for informed exploration of sparse models and does not require any additional tuning parameter. We rigorously characterise the stationary measure of the sticky Zig-Zag sampler.

    \item 
    We analyse the computational efficiency of the sticky Zig-Zag sampler by studying its complexity and mixing time.
    \item   
     We demonstrate the performance of the sticky Zig-Zag sampler on a variety of high dimensional statistical examples  
     (e.g. the example in Section~\ref{subsec: spatially structured sparsity} has dimensionality $10^6$). 
\end{itemize} 
The Julia package \texttt{ZigZagBoomerang.jl} (\cite{mschauer/ZigZagBoomerang.jl}) implements efficiently the sticky PDMP samplers from this article for general use. 

\subsection{Outline}

Section~\ref{sec: sticky PDMP samplers} formally introduces sticky PDMP samplers and gives the main theoretical results for the sticky Zig-Zag sampler. In  Section~\ref{sec: Subsampling} we explain how the sticky Zig-Zag sampler may be applied to subsampled data, allowing the algorithm to access only a fraction of data at each iteration, hence reducing the computational cost from $\cO(N)$ to $\cO(1)$, where $N$ is the sample size.
In Section~\ref{sec:runtime analysis} we extend the Gibbs sampler for variable selection for target measures of the form of Equation~\eqref{eq: target measure}. We analyse and compare the computational complexity and the mixing times of both the sticky Zig-Zag sampler and the Gibbs sampler. 
Section~\ref{sec: examples} presents four statistical examples with simulated data and analyses the outputs after applying the algorithms considered in this article. In Section~\ref{sec: discussions}  both limitations and  promising research directions are discussed. 

There are five appendices. 
The derivation of our theoretical results is given in Appendix~\ref{app: details of the sticky zz sampler}. Appendix~\ref{app: other sticky samplers} extends some of the theoretical results for two other sticky samplers: the sticky version of the Bouncy particle sampler (\cite{2015arXiv151002451B}) and the Boomerang sampler (\cite{bierkens2020boomerang}),  the latter having Hamiltonian deterministic dynamics invariant to a prescribed Gaussian measure.  Appendix~\ref{app: Comparison between reversible jump PDMPs and sticky PDMPs} contains a self-contained discussion with heuristic arguments and simulations which highlight the differences between the sticky PDMPs and the method of \textcite{chevallier2020reversible}. Appendix~\ref{app: details of section 4} complements Section~\ref{sec:runtime analysis} with the details of the derivations of the main results and by presenting local implementations of the sticky Zig-Zag sampler that benefit of a sparse dependence structure between the coordinates of the target measure. Appendix~\ref{app: details of examples} contains some of the details of the numerical examples of Section~\ref{sec: examples}.

\subsection{Notation}

The $i$th element of the vector $x \in \RR^d$ is denoted by $x_i$. We denote  $x_{-i} := (x_1,x_2,\dots,x_{i-1}, x_{i+1},\dots,x_d) \in \RR^{d-1}$. Write
\[
\left(x[k\colon y]\right)_{i}:=\left\{\begin{array}{ll}x_{i} & i \neq k, \\ y & i=k.\end{array}\right. 
\]
and $[x]_A := (x_i)_{i \in A} \in \RR^{|A|}$ for a set of indices $A\subset \{1,2,\dots,d\}$ with cardinality $|A|$. 
We denote by $\sqcup$ the disjoint union between sets and the positive and negative part of a real-valued function  $f$ by $f^+ := \max(0, f)$ and $f^- := \max(0, -f)$ respectively so that $f = f^+ - f^-$. For a topological space $E$, let $\cB(E)$ denote the Borel $\sigma$-algebra on $E$. Denote by $\cM(E)$ the class of Borel measurable functions $f\colon E \to \RR$ and let  $C(E) = \{f \in \cM(E)\colon f \text{ is  continuous }\}$. For a measure $\mu(\dd x, \dd y)$ on a product space $\cX,\cY$, we write the marginal measure on $\cX$ by $\mu(\dd x) = \int_{\cY} \mu(\dd x, \dd y)$.
\section{Sticky PDMP samplers}\label{sec: sticky PDMP samplers}
In what follows, we formally describe the sticky PDMP samplers (Section~\ref{sec: general framework}) and give the main theoretical results obtained for the sticky Zig-Zag sampler (Section~\ref{sec: sticky PDMP samplers>Sticky Zig-Zag sampler}). Section~\ref{sec: Subsampling} extends the sticky Zig-Zag sampler with subsampling methods.

\subsection{Construction of sticky PDMP samplers}\label{sec: general framework}  
The state space of the the sticky PDMPs contains two copies of zero for each coordinate position. This construction allows a coordinate process arriving at zero from below (or above) to spend an exponentially distributed  time at zero before jumping to the ``other'' zero and continuing the dynamics.  Formally, let $\overline \RR$ be the  disjoint union
$ \overline \R  = (-\infty,0^-] \sqcup [0^+,\infty)$
with the natural topology\footnote{{A function $f\colon \overline{\RR} \to \RR$ is continuous if both restrictions to $(\infty, 0^-]$ and $[0^+, \infty)$ are continuous. If $f(0^-) = f(0^+)$, we write $f(0)$.}} $\tau$,
where we use the notation $0^-$, $0^+$ to distinguish the zero element in $(-\infty,0]$ from the zero element in $[0,\infty)$.
The process has \emph{c\`adl\`ag}\footnote{I.e., trajectories that are continuous from the right, with existing limits from the left.} trajectories in the locally compact state space  $E = \overline\RR^d \times \cV$, where $\cV \subset \RR^d$. 
Pairs of position and velocity will typically be denoted by  $ (x, v) \in \overline\RR^d \times \cV$.
A trajectory reaching zero in a coordinate from below (with positive velocity) or from above (with negative velocity) spends time at the closed end of the half open interval $(-\infty, 0^-]$ or $[0^+, \infty)$, respectively. For $i = 1,\dots, d$ we define the associated `frozen boundary' $\mathfrak F_i \subset E$ for the $i$th coordinate as 
\[\mathfrak F_i :=\{(x,v) \in E\colon x_i = 0^-, \, v_i >0  \, \text{ or } \, x_i = 0^+, \, v_i < 0\} .\] 
Thus the $i$th coordinate of the particle is sticking to zero (or frozen), if the state of the particle belongs to the $i$th frozen boundary $\mathfrak F_i$. 

Sometimes, we abuse notation by writing $(x_i,v_i) \in \mathfrak F_i$ when $(x,v) \in \mathfrak F_i$ as the set $\mathfrak F_i$ has restrictions only on $x_i, v_i$.  
The closed endpoints of the half-open intervals are somewhat reminiscent of sticky boundaries in the sense of \textcite[Example 5.59]{Liggett2010}.
Denote by $\alpha \equiv  \alpha(x,v)$ the set of indices of active coordinates corresponding to state $(x,v)$, defined by 
\begin{equation}\label{eq:alpha}
\alpha(x,v)= \{i\in \{1,2,\dots, d\} \colon (x, v) \notin \mathfrak F_i \}  \end{equation}
and its complement $\alpha^c = \{1,2,\dots,d\}\setminus 
\alpha$.
Furthermore define a jump or \emph{transfer mapping} $T_i \colon \mathfrak F_i \rightarrow E$ by 
\[ T_i (x,v) = \begin{cases}
(x[i\colon 0^+], v) & \text{if } x_i = 0^-, v_i > 0,\\
(x[i\colon 0^-], v) & \text{if } x_i = 0^+, v_i < 0.
\end{cases}\]
The sticky PDMPs on the space $E$ are determined by their infinitesimal characteristics:
their dynamics are determined by random state changes happening at random jump times of a time inhomogeneous Poisson process with intensity depending on the state of the process, and a deterministic flow governed by a differential equation in between.
The state changes are characterised by a Markov kernel $\cQ\colon E\times \cB(E) \to [0,1]$,
at random times sampled with state dependent intensity $\lambda\colon E \to [0,\infty)$.
The deterministic dynamics are determined coordinate-wise by the integral equation 
\begin{equation}
    \label{eq: deterministic dynamics}
    (x_i(t), v_i(t)) = (x_i(s),v_i(s)) +\int_s^t \xi_i(x_i(r), v_i(r)) \dd r, \quad i=1,2,\dots,d,
\end{equation}
with $\xi_i$ being state dependent with form
\begin{equation}
\label{eq: deterministic dynamics 2}
\xi_i(x,v) = \begin{cases}
\bar\xi_i(x_i,v_i) & (x_i,v_i) \notin \fF_i\\
(0, 0) & (x_i,v_i) \in \fF_i,
\end{cases}
\end{equation}

 for  functions $\bar \xi_i\colon \overline \RR \times \RR\to\overline \RR \times \RR$ which depend on the specific PDMP chosen and corresponds to the coordinate-wise dynamics  of the ordinary PDMP while the  second case in Equation~\eqref{eq: deterministic dynamics 2} captures the behaviour of the $i$th coordinate when it sticks at 0.

For PDMP samplers, we typically have $\bar \xi_i = \bar \xi_j$ for all $i,j \in 1, \dots, d$ and we have different types of state changes given by Markov kernels $\cQ_1$, $\cQ_2$, \dots, for example refreshments of the velocity, reflections of the velocity, unfreezing of a coordinate etc.  If each transition is triggered by its individual independent Poisson clock with intensity $\lambda_1, \lambda_2, \ldots$, then  $\lambda = \sum_i \lambda_i$, and $\cQ$ itself can be written as the mixture \[\cQ((x,v), \cdot) = \sum_i \frac{\lambda_i((x,v))}{\lambda((x,v))} \cQ_i((x,v), \cdot).\]

With that, the dynamics of the sticky PDMP sampler $t\mapsto (X(t), V(t))$ are as follows: starting from $(x, v) \in E$, 
    \begin{enumerate}
        \item its flow in each coordinate is deterministic and continuous until an event happens. The deterministic dynamics are given by \eqref{eq: deterministic dynamics}. Upon hitting $\mathfrak F_i$, the $i$th coordinate process freezes, captured by the state dependence of \eqref{eq: deterministic dynamics 2}.
        \item   A frozen coordinate  ``unfreezes'' or ``thaws'' at rate equal to $\kappa_i|v_i|$ by jumping according to the transfer mapping $T_i$ to the location $(0^+, v_i)$ (or $(0^-, v_i)$) outside $\fF_i$ and continuing with the \emph{same} velocity as before.
        That is, on hitting $\mathfrak F_i$, the $i$th coordinate process freezes for an independent exponentially distributed time with rate $\kappa_i|v_i|$.
             This constitutes a non-reversible move between models of different dimension. 
       The corresponding transition $\cQ_{i,\operatorname{thaw}}$ is the Dirac measure at $ \delta_{T_i(x,v)}$ and the intensity component $\lambda_{i,\operatorname{thaw}}$ equals $\kappa_i|v_i| \ind_{\fF_i}$. 
        \item An inhomogeneous Poisson process $\lambda_{\mathrm{refl}}$ with  rate depending on  $\Psi$ triggers the reflection events. At a reflection event time, the process changes its velocities according to its reflection rule $\cQ_{\mathrm{refl}}$ in such a way that the process is invariant to the measure $\mu$. 
        \item  Refreshment events can be added, where, at exponentially distributed  inter-arrival times, the velocity changes according to a refreshment rule leaving the measure $\mu$ invariant. Refreshments are sometimes necessary for the process to be ergodic.  
        \end{enumerate}
        
The resulting stochastic process $(X_t, V_t)$ is a sticky PDMP with dynamics $\cQ$, $\lambda$, $\varphi$, initialised in  $(X(\tau_0),V(\tau_0))$.
Let $s \to \varphi(s, x, v)$  be the deterministic solution of \eqref{eq: deterministic dynamics}  starting in $(x,v)$. Set $\tau_0=0$ and the initial state $(X(\tau_0),V(\tau_0)) \in E$. A sample of a sticky PDMP is given by the recursive construction in Algorithm~\ref{alg: algorithm 1}. 
  
\begin{algorithm}
\caption{PDMP samplers: recursive construction}
Given the current state $(X(\tau_k), V(\tau_k))$ at time $\tau_k$
\begin{enumerate}
    \item Sample independently $\Delta_k$ as the first event time of an inhomogeneous Poisson process. We denote $\Delta_k \sim \text{Poiss}(s \to \lambda(\varphi(s, X(\tau_k), V(\tau_k)))$,  with
\begin{equation}\label{eq:jumptimes} \PP\left( \Delta_k \ge t\right) = \exp\left(-\int_0^t  \lambda(\varphi(s, X(\tau_k), V(\tau_k)) \dd s\right). \end{equation} 
\item Let $\tau_{k+1} = \tau_k + \Delta_k$ and set for $t \in [\tau_k, \tau_{k+1})$ 
\[
(X(t), V(t)) = \varphi(t - \tau_k, X(\tau_k), V(\tau_k)).
\]
\item Let 
\[(X(\tau_{k+1}), V(\tau_{k+1})) \sim \cQ(\varphi(\Delta_k, X(\tau_k), V(\tau_k)), \cdot).\]
\end{enumerate}
\label{alg: algorithm 1}
\end{algorithm}

In what follows, we focus our attention on the Sticky Zig-Zag sampler and defer to Appendix~\ref{app: other sticky samplers} the details of the Bouncy Particle sampler and the Boomerang samplers.

\subsection{Sticky Zig-Zag sampler}
A trajectory of the Sticky Zig-Zag sampler  has piecewise constant velocity which is an  element of the set $\mathcal{V} = \{v \colon |v_i| = a_i, \forall i \in \{1,2,\dots, d\}\}$
for a fixed vector $a$. For each index $i$,  the deterministic dynamics of Equation~\eqref{eq: deterministic dynamics 2} are determined by the function $\bar \xi_i(x_i, v_i) = (v_i, 0)$. 
The reflection 
rate $\lambda_{\mathrm{refl}}$ 
is \emph{factorised} coordinate-wise and the reflection event for the $i$th coordinate is determined by the inhomogeneous rate \begin{equation}
    \label{eq: poisson rate zigzag}
     \lambda_{i, \mathrm{refl}}(x,v) = \1_{i\in \alpha(x,v)}(v_i \partial_i \Psi(x))^+.
\end{equation}

At reflection time of the $i$th coordinate, the transition kernel $\cQ_{i, \mathrm{refl}}$ acts deterministically by flipping the sign of the $i$th velocity component of the state: $(x_i, v_i) \to (x_i, -v_i)$. As shown in \textcite{bierkens2019ergodicity}, the Zig-Zag sampler does not require refreshment events in general to be ergodic.

\subsection{Theoretical aspects of the Sticky Zig-Zag sampler}\label{sec: sticky PDMP samplers>Sticky Zig-Zag sampler}
A theoretical analysis of the sticky Zig-Zag sampler is given in Appendix~\ref{app : construction sticky zz}. In this section we review key concepts and state the main results.

The stationary measure of a PDMP is studied by looking at the extended generator of the process which is an operator characterising the process in terms of local martingales - see \textcite[Section 14 ]{Davis1993} for details. The extended generator is - as the name suggests - an extension of the infinitesimal generator of the process (defined for example in \cite[Theorem 3.16]{Liggett2010}) in the sense that it acts on a larger class of functions than the infinitesimal generator and it coincides with the infinitesimal generator when applied to functions in the domain of the infinitesimal generator.

A general representation of the extended generator of PDMPs is given in \textcite[Section 26]{Davis1993}, while the infinitesimal generator of the ordinary Zig-Zag sampler is given in the supplementary material of \textcite{bierkens2019}. Here, we highlight the main results we have derived for the sticky Zig-Zag sampler.

 Recall  $t \to \varphi(t, x, v)$ denotes the  deterministic solution of \eqref{eq: deterministic dynamics} starting in $(x,v)$ and $\tau$ is the natural topology on $E$.
Define the operator $\cA$ with domain 
\begin{align*}
    \cD(\cA) = \{f &\in \cM(E): \, t \mapsto f(\varphi(t,x, v)) \text{ $\tau$-absolutely continuous } \forall (x,v) \ \text{and} \\
    & \forall i: \,\lim_{t \downarrow 0} f(x[i\colon 0^+ + t], \cdot) = f(x[i\colon 0^+], \cdot ),\, \lim_{t \downarrow 0} f(x[i\colon 0^- - t], \cdot ) = f(x[i\colon 0^-], \cdot)\}
\end{align*}
by $\cA f(x, v) = \sum_{i=1}^d \cA_i f(x, v)$
with
\[
\cA_i f(x, v) = \begin{cases}
a_i \kappa_i \left(f(T_i(x, v)) - f(x, v)\right) & (x, v) \in \mathfrak{F}_i,\\
v_i \partial_{x_i} f(x, v) + \lambda_i(x,v)\left(f(x, v[i:-v_i]) - f(x, v)\right) & \text{else.}
\end{cases}
\]

\begin{proposition}\label{prop: gen and ext gen of sticky zz}
The extended generator of the $d$-dimensional Sticky Zig-Zag process is given by $\mathcal A$ with domain $\mathcal D(\mathcal A)$. 
\end{proposition}
\begin{proof}
See Appendix~\ref{app: the extended generator of sticky zig zag}.
\end{proof}
Notice that, the operator $\cA$ restricted on $D = \{f \in C^1_c( E), \cA f \in C_b( E)\}$ coincides with the infinitesiaml generator of the ordinary Zig-Zag process restricted on $D$, see Proposition~\ref{prop: ext generator restricted resable standard zigzag generator}, Appendix~\ref{app: the extended generator of sticky zig zag} for details.

\begin{theorem}\label{thm: Sticky Zig-Zag main result}
The $d$-dimensional Sticky Zig-Zag sampler is a Feller process and a strong Markov process  in the topological space $(E, \tau)$ with stationary measure
\begin{equation}\label{eq: Sticky Zig-Zag invariant measure}
    \mu(\dd x, \dd v) =    \frac1C \sum_{u \in \mathcal{V}} \exp(-\Psi(x))\prod_{i = 1}^d\left(\dd x_i + \frac{1}{\kappa_i}\left( \1_{v_i>0}\, \delta_{0^-}(\dd x_i) + \1_{v_i<0}\, \delta_{0^+}(\dd x_i) \right) \delta_{u}(\dd v)\right),   
\end{equation}
for some normalization constant $C>0$.
\end{theorem}
\begin{proof} The construction of the process and the characterization of the extended generator and its domain of the $d$-dimensional Sticky Zig-Zag process can be found in Appendix~\ref{app : construction sticky zz}. We then prove that the process is Feller and strong Markov (Appendix~\ref{app: strong markov property sticky zz} and Appendix~\ref{app: feller property sticky zz}). By \textcite[Theorem~3.37]{Liggett2010}, $\mu$ is a stationary measure if, for all $f \in D$, $\int \cL f \dd \mu = 0$. This last equality is derived in Appendix~\ref{app: remaining proofs of section 2}. 
\end{proof}
\begin{theorem}\label{thm: Ergodicity of szz}
    Suppose $\Psi$ satisfies Assumption~\ref{ass: Assumption ergodicity}. Then the sticky Zig-Zag process is ergodic and $\mu$ is its unique stationary measure.
\end{theorem}
\begin{proof} See Appendix~\ref{app: ergodicity of the szz}. 
\end{proof}


The following remark establishes a formula for the recurrence time of the Sticky Zig-Zag to the null model, and may serve as guidance in design of the probabilistic model or the choice of the  parameter $\kappa_i$, here assumed  for simplicity to be all equal. 
\begin{remark} \label{rmk: recurrence time}
    \emph{(Recurrence time of the Sticky Zig-Zag to zero)} The expected time to leave the position $\bm0 = (0,0,\dots,0)$ for a $d$-dimensional  Sticky Zig-Zag with unit velocity components is $\frac{1}{\kappa d}$ (since each coordinate leaves 0 according to an exponential random variable with parameter $\kappa$). 
    A simple argument given in Appendix~\ref{app: recurrence time to 0} shows that 
    the expected time of the process to return to the null model is 
    \begin{equation}
        \label{eq: recurrence time}
    \frac{1- \mu(\{\bm 0 \})}{d\kappa \mu(\{\bm 0\})}.
    \end{equation}
     \end{remark}

\subsection{Extension: sticky Zig-Zag sampler with subsampling method}\label{sec: Subsampling}
Here we  address the problem of sampling a $d$-dimensional target measure when the log-likelihood is a sum of $N$ terms, when $d$ and $N$ are large. Consider for example a regression problem where both the number of covariates and the number of experimental units in the  dataset are large. In this situation  full evaluation of the log-likelihood and its gradient is prohibitive. However, PDMP samplers can still be used with the exact subsampling technique (e.g.~\cite{bierkens2019}) as this allows for substituting the  gradient of the log-likelihood (which is required for deriving the reflection times) by an estimate of it which is cheaper to evaluate, without introducing any bias on the output of the sampler.

The subsampling technique for Sticky Zig-Zag samplers requires to find an unbiased estimate of the gradient of $\Psi$ in  \eqref{eq: target measure}. To that end, assume the following decomposition:
\begin{equation}
    \label{eq: decomposition partial derivatives}
    \partial_{x_i} \Psi(x) =  \left(\sum_{j = 1}^{N_i} S(x, i, j)\right), \quad \forall x \in \overline \RR^d, \, i = 1,2,\dots, d, 
\end{equation}
for some scalar valued function $S$. This assumption on $\Psi$ is satisfied for example for the setting with a spike-and-slab prior and a likelihood that is a product of factors, such as for likelihoods of (conditionally) independent observations.

For fixed $(x,v)$ and $x^* \in \RR^d$, for each $i \in \alpha(x,v)$ the random variable \[N_i\left(S (x, i, J) - S (x^*, i, J )\right) + \partial_{x_i} \Psi(x^*), \quad J \sim \text{Unif}(\{1,2,\dots,N_i\})\] is an unbiased estimator for $\partial_{x_i} \Psi(x)$. Define the Poisson rates \[\tilde \lambda_{i,j}(x, v) = \left(v_i N_i(S(x, i, j) - S(x^*, i, j)) + v_i\partial_{x_i}\Psi(x^*)\right)^+\] and, for each $i \in \alpha$, define the bounding rate
\begin{equation*}
\overline \lambda_i(t, x, v) \ge \tilde \lambda_{i, j}(\varphi(t, x, v)), \quad t\ge 0, \,\forall j \in \{1,2,\dots,N_i\}, 
\end{equation*}
which is specified by the user and such that Poisson times with inhomogeneous rate $\tau \sim \text{Poiss}(s\to \overline \lambda_i(s,x,v))$ can be simulated (see Appendix~\ref{app: simulating sticky PDMPs} for details on the simulation of Poisson times).

The Sticky Zig-Zag with subsampling has the following dynamics:
\begin{itemize}
    \item the deterministic dynamics and the sticky events are identical to the ones of the Sticky Zig-Zag sampler presented in Section~\ref{sec: sticky PDMP samplers>Sticky Zig-Zag sampler};
    \item a \emph{proposed} reflection time equals $\min_{i\in \alpha(x,v)} \tau_i$, with $\{\tau_i\}_{i\in \alpha(x,v)}$ being independent inhomogeneous Poisson times with rates   $s \to \overline \lambda_i(s,x,v)$;
    \item at the proposed reflection time $\tau$ triggered by the $i$th Poisson clock, the process reflects its velocity according to the rule $(x,v) \to(x,v[i, -v_i])$ with probability $\tilde \lambda_{i, J}(\varphi(\tau, x, v))/\overline \lambda_i(\tau, x, v)$ where $J \sim \text{Unif}(\{1,2,\dots,N_i\})$.
\end{itemize}
 
\begin{proposition}\label{prop: invariant subsampler}
  The Sticky Zig-Zag with subsampling has a unique stationary measure given by Equation~\eqref{eq: Sticky Zig-Zag invariant measure}.
\end{proposition}
The proof of Proposition~\ref{prop: invariant subsampler} follows with a similar argument made in the proof of \textcite[Theorem~4.1]{bierkens2019}. The number of computations required by the Sticky Zig-Zag with subsampling to compute the next event time with respect to the quantity $N$ is  $\mathcal{O}(1)$ (since $\partial_{x_i} \Psi(x^*)$ can be pre-computed). This advantage comes at the cost of introducing `shadow event times', which are event times where the velocity component does not reflect. In case the posterior density satisfies a Bernstein-von-Mises theorem, the advantage of using subsampling over the standard samplers has been empirically shown and informally argued for in  \textcite[Section 5]{bierkens2019} and \textcite[Section 3]{bierkens2020boomerang} for large $N$ and when choosing $x^*$ to be the mode of the posterior density.

\section{Performance comparisons for Gaussian models}\label{sec:runtime analysis}
In this section we discuss the performance of the Sticky Zig-Zag sampler in comparison with a Gibbs sampler.
The sticky Zig-Zag sampler includes new coordinates randomly but uses gradient information to find which coordinates are zero. By comparing to a Gibbs sampler that just  proposes models at random, we show that it is an efficient scheme of exploration. As the Gibbs sampler requires closed form expression of Bayes factors between different (sub-)models (Equation~\eqref{eq; Bayes_factors} below), we consider Gaussian models.
The comparison is motivated by considering two samplers that do not require model specific proposals or other tuning parameters.
In specific cases such as the target models considered below, the Gibbs sampler could be improved  by carefully choosing a problem-specific proposal kernel in between (sub-)models, see for example \textcite{zanella2018scalable} and \textcite{liang2021adaptive} -- something we don't consider here.

The comparison is primarily
in relation to the dimension $d$, average number of active particles and sample size $N$ of the problem.  It is well known that the performance of a Markov chain Monte Carlo method is given by both the computational cost of simulating the algorithm and the convergence properties of the underlying process. In Section~\ref{sec: main results} we consider both these aspects and compare the results obtained for the sticky Zig-Zag sampler with those relative to the Gibbs sampler. The results are summarised in Table~\ref{tab: 1} and Table~\ref{tab: 2}. The technical details of this section are given in Appendix~\ref{app: details of section 4}. 

\subsection{Gibbs sampler}\label{sec: gibbs sampler}
We can use a set of active indices $\alpha$ to define a model, as the corresponding set of non-zero values in $\R^d$:
 \[ \cM_\alpha := \{x\in \RR^d \colon x_i = 0, i \notin \alpha\} \quad \text{for $\alpha \subset \{1,2,\dots,d\}$.}\]
For every set of indices $\alpha \subset \{1,2,\dots,d\} $ and for every $j$, the Bayes factors relative to two neighbouring (sub-)models  (those differing by only one coefficient) for a measure as in Equation~\eqref{eq: target measure} are given by
 \begin{equation}\label{eq; Bayes_factors}
 B_j(\alpha) =  \frac{\mu( \cM_{\alpha \cup \{j\}})}{\mu(
\cM_{\alpha \setminus\{ j\}})} = \frac{\kappa_j\int_{\RR^{|\alpha \cup \{j\}|}} \exp(-\Psi(y)) \dd x_{\alpha \cup \{j\} }}{\int_{\RR^{|\alpha \setminus \{j\}|}} \exp(-\Psi(z)) \dd x_{\alpha \setminus \{j\}}}, 
 \end{equation}
   where $y = \{x \in \RR^d \colon x_i = 0,\, i \notin (\alpha \cup \{j\}) \}$, $z = \{x \in \RR^d \colon x_i = 0, \, i \notin (\alpha \setminus \{j\})]$.
 The Gibbs sampler starting in $(x, \alpha)$, with $x_i \ne 0$ only if $i \in \alpha$ for some set of indices $\alpha \subset \{1,2,\dots,d\} $, iterates the following two steps: 
\begin{enumerate}
    \item Update $\alpha$ by choosing randomly $j \sim \text{Unif}(\{1,2,\dots,d\})$ and set $\alpha \leftarrow \alpha \cup \{j\}$  with probability $p_j$ where $p_j$ satisfies 
    $p_j/(1-p_j) = B_j(\alpha)$, 
        otherwise set $\alpha \leftarrow \alpha \setminus \{j\}$.
    \item Update the free coefficients $x_\alpha$ according to the marginal probability of $x_{\alpha}$ conditioned on $x_i = 0$ for all $i\in \alpha^c$.
\end{enumerate} 
In Appendix~\ref{app: bayes_factors}, we give an analytical expressions for the right hand-side of Equation~\eqref{eq; Bayes_factors} and the conditional probability in step 2 when $\Psi$ is a quadratic function of $x$. For logistic regression models, neither step 1 nor step 2 can be directly derived and the Gibbs samplers makes use of a further auxiliary P\'{o}lya-Gamma random variable $\omega$ which has to be simulated at every iteration and makes the computations of step 1 and step 2 tractable, conditionally on $\omega$ (see \cite{polson2013bayesian} for details). 

\subsection{Runtime analysis and mixing times}\label{sec: main results}

The ordinary Zig-Zag sampler can greatly profit in the case of models with a sparse conditional dependence structure between coordinates by employing local versions of the standard algorithm as presented in \textcite{bierkens2020piecewise}. In Appendix~\ref{app: simulating sticky PDMPs} we discuss how to simulate sticky PDMPs and derive similar local algorithms relative to the sticky Zig-Zag. Also the Gibbs sampler algorithm, as described in Section~\ref{sec: gibbs sampler}, benefits  when the conditional dependence structure of the target is sparse. In Appendix~\ref{app: runtimes of the algorithms} we analyse the computational complexity of both algorithms. In the analysis, we drop the dependence on $(x,v)$  and we assume that the size of $\alpha(t) := \{i \colon x_i(t) \ne 0\}$ fluctuates around a typical value $p$ in stationarity. Thus $p$ represents the number of non-zero components in a typical model, and can be much smaller than $d$ in sparse models.

Table~\ref{tab: 1} summarises the results obtained of both algorithms in terms of the sample size $N$ and $p$ when the conditional dependence structure between the coordinates of the target is full and the sub-sampling method presented in Section~\ref{sec: Subsampling} cannot be employed (left-column) and when there is sparse dependence structure and subsampling can be employed (right-column). Our findings are validated by numerical experiments in Section~\ref{sec: examples} (Figure~\ref{fig: runtime of heart}, Figure~\ref{fig:logistic_comparison}). 
 
\begin{table}[ht!]
\centering
\begin{tabular}{ c c c  } 
 Algorithm & Worst case & Best case \\ \hline
 Sticky Zig-Zag &
$p^2 N$
 & 
$p$
 \\ 
 Gibbs sampler & $p(p^2 + N)$ & $p(\sqrt{p} + N)$ \\ 
 \hline
\end{tabular}

\caption{Computational scaling of the Sticky Zig-Zag algorithm and the Gibbs sampler for variable selection for $p$ and sample size $N$. Worst case is when the target density does not present any conditional independence structure and the subsampling method for the Sticky Zig-Zag cannot be employed; best case when the target measure presents a relevant conditional independence structure and subsampling can be employed.}
\label{tab: 1}
\end{table}

We now turn our focus on the mixing time of both the underlying processes. Given the different nature of dependencies of the two algorithms, a rigorous and theoretical comparison of their mixing times is difficult and outside the scope of this work. We therefore provide an heuristic argument for two specific scenarios where we let both algorithms be initialized at $x \sim \cN_d(0, I) \in \RR^d$, hence in the full model, and assume that the target $\mu$ assigns most of its probability mass to the null model $\cM_{\emptyset}$.
Then we derive the expected hitting time to $\cM_{\emptyset}$ for both processes. 
The two scenarios differ as in the former case the target $\mu$ is supported in every sub-model so that the process can reach the point $(0,0,\dots,0)$ by visiting any sequence of sub-models while in the latter case the measure $\mu$ is supported in a single nested sequence of sub-models. Details of the two scenarios are given in Appendix~\ref{app: mixing}. Table~\ref{tab: 2} summarizes the scaling results (in terms of dimensions $d$) derived in the two cases considered. 
\begin{table}[ht!]
\centering
\begin{tabular}{ c c c  } 
 Algorithm & $\mu$ supported on every model & $\mu$ supported on a nested sequence \\ \hline
 Sticky Zig-Zag & $\log(d)$ & $d$  \\ 
 Gibbs sampler & $d\log(d)$ & $d^2$ \\ 
 \hline
\end{tabular} q
\caption{Scaling relative to the dimension $d$ of the expected time (number of iteration for the Gibbs sampler) to travel from the full model (initialized as a standard Gaussian random variable) to the null model (which is the mode of the target). The results are for targets which are supported in every model and for targets supported on a single sequence of nested sub-models.}
\label{tab: 2}
\end{table}

\section{Examples}\label{sec: examples}
In this section 
 we apply the Sticky Zig-Zag sampler and, when possible, compare its performance with the Gibbs sampler in four different problems of varying nature and difficulty:
\begin{itemize}
    \item[\ref{subse: learning networks}] \emph{(Learning networks of stochastic differential equations)} A system of interacting agents where the dynamics of each agent are given by a stochastic differential equation. We aim to infer the interactions among agents. This is an example where the likelihood does not factorise and the number of parameters increases quadratically with the number of agents. We demonstrate the Sticky Zig-Zag sampler under a spike-and-slab prior on the parameters that govern the interaction and compare this with the Gibbs sampler.
    \item[\ref{sec: example>Spatially structured sparsity}] \emph{(Spatially structured sparsity)} An image denoising problem where the prior incorporates  that a large part of the image is black (corresponding to sparsity), but also promotes positive correlation among neighbouring pixels. Specifically, this examples illustrates that the Sticky Zig-Zag sampler can be employed  in high dimensional regimes (the showcase is in dimension one million) and for sparsity promoting priors other than factorised priors such as  spike-and-slab priors. 
    \item[\ref{subsec: examples>logistic regression}] \emph{(Logistic regression)} The logistic regression model where both the number of covariates  and the sample size are large, while assuming the coefficient vector to be sparse.  This is an non-Gaussian optimal scenario where the Sticky Zig-Zag sampler can be employed with subsampling technique achieving $\cO(1)$ scaling with respect to the sample size. 
    \item[\ref{subsec: sparse precision matrix}] \emph{(Estimating a sparse precision matrix)} The setting where $N$ realisations of independent  Gaussian vectors with  precision matrix of the form $ X X'$ are observed. Sparsity is assumed on the off-diagonal elements of the lower-triangular matrix $X$. What makes this example particularly interesting is that the gradient of the log-likelihood explodes in some hyper-planes, complicating the application of  gradient-based Markov chain Monte Carlo methods.  
\end{itemize}
In all cases we simulate data from the model and assume the parameter to be sparse (i.e.\ most of its elements are assumed to be zero) and high dimensional. In case a spike-and-slab prior is used, the slabs are always chosen to be zero-mean Gaussian with (large) variance $\sigma_0^2$. 
The sample sizes, parameter dimensions and additional difficulties such as correlated parameters or non-linearities  which are considered in this section illustrate the computational efficiency of our method (and implementation) in a wide range of settings.   In all examples we used either the local or the fully local algorithm of the Sticky Zig-Zag as detailed in Appendix~\ref{app: simulating sticky PDMPs} with velocities in the set $\cV = \{-1,+1\}^d$. Comparisons with the Gibbs sampler are  possible for Gaussian models and the logistic regression model. Our implementation of the Gibbs sampler is taking advantage of model sparsity. Because of its computational overhead, when such comparisons are included,  the dimensionality of the problems considered has been reduced. The performance of  the two algorithms is compared by running the two algorithms for approximately the same computing time. As performance measure we consider the squared error as a function of the computing time:
\begin{equation}
    \label{eq: squared error}
    c \mapsto \cE_{\text{s}}(c) := \sum_{i=1}^d (p^{\text{s}}_i(c) - \overline{p}_i)^2,
\end{equation}
where $c$ denotes computing time (we use $c$ rather than $t$ as the latter is used as time index for the Zig-Zag sampler).  
In the displayed expression, we first compute $\overline{p}_i$, which is an approximation to the posterior probability of the $i$th coordinate being nonzero. This quantity can either be obtained by running the Sticky Zig-Zag sampler or the Gibbs sampler (if applicable) for a very long time. As we show the Sticky Zig-Zag sampler to converge faster, especially in high dimensional problems, we use this sampler in approximating this value. We stress that the same result could be obtained by running the Gibbs sampler for a very long time.  More precisely, we compute for each coordinate of the Sticky Zig-Zag sampler the fraction of time  it is nonzero. In $\cE_{\text{s}}(c)$, the value of $\overline{p}_i$ is compared to $p^{\text{s}}_i(c)$ which is 
 the fraction of time (or fraction of samples in case of the Gibbs sampler) where $x_i$ is nonzero  using computational budget $c$ and sampler `s'.  All the experiments were carried out with a conventional laptop with Intel core i5-10310 processor and 16GB DDR4 RAM. Pre-processing time and memory allocation of both algorithms are comparable.

\subsection{Learning networks of stochastic differential equations.}\label{subse: learning networks}
In this example we consider a stochastic model for  $p$ autonomously moving agents (``boids'') in the plane. The dynamics of the   location of the $i$th agent is assumed to satisfy the stochastic differential equation
\begin{equation}\label{eq: agent i}
\dd U_i(s) = -\lambda U_i(s)\dd s + \sum_{j \ne i} x_{i,j}(U_{j}(s) - U_{i}(s)) \dd s + \sigma \dd W_i(s), \qquad 1\le i \le p 
\end{equation}
where, for each $i$, $(W_i(s))_{0\le s\le T}$ is an independent $2$-dimensional Wiener process. 
We assume the trajectory of each agent is observed continuously over a fixed interval $[0,T]$. This implies $\sigma>0$ can be considered known, as it can be recovered without error from the quadratic variation of the observed path. For simplicity we will also assume the mean-reversion parameter $\lambda>0$ to be known.  Let  $x = \{x_{i,j} \colon i \ne j\} \in \RR^{p^2 - p}$ denote the unknown parameter.
  If $x_{i,j} > 0$, agent $i$  has the tendency to follow agent $j$, on the other hand, if $x_{i,j} < 0$, agent $i$ tends to avoid agent $j$. Hence, estimation of $x$ aims at inferring which agent follows/avoids other agents. We will study this problem from a Bayesian point of view assuming sparsity of $x$, incorporated via the prior using a spike and slab prior.  This problem has been studied previously  in  \textcite{bento2010learning} using $\ell_1$-regularised least squares estimation.

  Motivation for studying this problem can be found in  \textcite{10.1145/37401.37406} and the presentation at \textcite{dancing}.  An  animation of the trajectories of the agents in time can be found at \textcite{boidanimation}.

Suppose $U_i(s)=(U_{i,1}(s), U_{i,2}(s))$ and let 
$Y(s) = (U_{1,1}(s),\ldots, U_{p,1}(s), U_{1,2}(s),\ldots, U_{p,2}(s))$ 
denote the vector obtained upon concatenation of all $x$-coordinates and $y$-coordinates of all agents. 
Then, it follows from Equation~\eqref{eq: agent i} that $\dd Y(s) = C(x) Y(s) \dd s + \sigma \dd W(s)$, 
where $W(s)$ is a Wiener process in $\RR^{2p}$. Here,  $C(x)=\mbox{diag}(A(x), A(x))$ where
\[
A(x) = \begin{bmatrix}
    -\lambda -  \overline{x}_1  & x_{1,2} & x_{1,3} & \dots\\
    x_{2,1} & -\lambda - \overline{x}_2 &  x_{2,3} & \\
    x_{3,1} & & \ddots & \\
    \vdots & & & 
\end{bmatrix}
\]
with $\overline{x}_i = \sum_{j \ne i} x_{i,j}$.
If $\mathbb{P}_x$ denotes the measure on path space of $Y_T:=(Y(s),\, s\in [0,T])$ and $\mathbb{P}_0$ denotes the Wiener-measure on $\RR^{2p}$, then it follows from Girsanov's theorem that 
\begin{equation}
    \label{eq: log-girsanov}
\ell(x):= \log \frac{\mathbb{P}_x}{\mathbb{P}_0}(Y_T)   =   \frac{1}{\sigma^{2}}\int_0^T (C(x) Y(s))' \dd Y(s) - \frac{1}{2\sigma^{2}} \int_0^T \|C(x) Y(s) \|^2 \dd s.
\end{equation}
As we will numerically only be able to store the observed sample path on a fine grid, we approximate the  integrals appearing in the  log-likelihood $\ell(x)$ using a standard Riemann-sum approximation of It\^o integrals (see e.g.\ \cite[Ch. IV, sec. 47]{rogers2000diffusions2}) and time integrals. We assume $x$ to be sparse which is incorporated by choosing a spike-and-slab  prior for $x$ as in Equation~\eqref{eq: spike-and-slab prior}. The posterior measure is of the form of \eqref{eq: target measure} with $\kappa$ and $\Psi(x)$ as in \eqref{eq: equivalence spike-and-slab}. 
  As $x\mapsto \Psi(x)$ is quadratic, the reflection times of the Sticky Zig-Zag sampler can be computed in closed form.

\medskip 

\textbf{Numerical experiments:}
In our numerical experiments we  fix $p = 50$ (number of agents), $T = 200$ (length of time-interval), $\sigma = 0.1$ (noise-level) and $\lambda = 0.2$ (mean-reversion coefficient). We set the parameter $x$ such that each agent has one agent that tends to follow and one agent that tends to avoid. Hence, for every $i$, we set $x_{i,j}$ to be zero for all $j \ne i$, except for 2 distinct indices $j_1,j_2 \sim \text{Unif}(\{1,2,\dots,d\}\setminus i)$ with $x_{i,j_1}x_{i,j_2}<0$. The parameter $x$ is very sparse and it is highly nontrivial to recover its value.  We then simulate $Y_T$ using  Euler forward discretization scheme, with step-size equal to $0.1$ and initial configuration $Y(0)\sim \cN_{2p} (0,I)$. 

The prior weights $w_1 = w_2 = \dots = w_d$ ($w_i$ being the prior probability of the $i$th coordinate to be nonzero) are conveniently chosen to  equal the proportion of non-zero elements in the true (data-generating) parameter vector $x$.  The variance of each slab was taken to be  $\sigma^2_0 = 50$.
We ran the Sticky Zig-Zag sampler with final clock $500$, where the algorithm was initialized in the full-model with no coordinate frozen at 0 at the posterior mean of the Gaussian density proportional to $\Psi$.

Figure~\ref{fig:boids} shows the discrepancy between the parameters used during simulation (ground truth) and the estimated posterior median.  In this figure, from the (sticky) Zig-Zag trajectory of each element $x_{i,j}$ ($i\neq j$) we  collected their values at time  $t_i =i 0.1 $ and  subsequently computed the median of the those values. We conclude that all parameters which are strictly positive (coloured in pink) are recovered well. At the bottom of the figure (black points and crosses), $25$ are incorrectly identified as either being zero or negative.  In this experiment, the Sticky Zig-Zag sampler outperforms the Gibbs sampler considerably.


In Figure~\ref{fig:boids_comparisons} we compare the performance of the Sticky Zig-Zag sampler with the Gibbs sampler. Here, all the parameters (including initialisation) are as above, except
now the number of agents is taken as $p = 20$. Both $c \mapsto \cE_{\text{Zig-Zag}}(c)$ and $c \mapsto \cE_{\text{Gibbs}}(c)$, with $c$ denoting the computational budget,  are computed for $c \in [0,10]$. For this, the final clock of the Zig-Zag was set to $10^4$ and the number of iterations for the  Gibbs sampler was set to $1.2\times10^4$. For obtaining $\bar{p}_i$ the Sticky Zig-Zag sampler was run  with  final clock $5\times10^4$ (taking approximately 50 seconds computing time). 

\begin{figure}[ht!]
    \centering
    \includegraphics[width=0.7\linewidth]{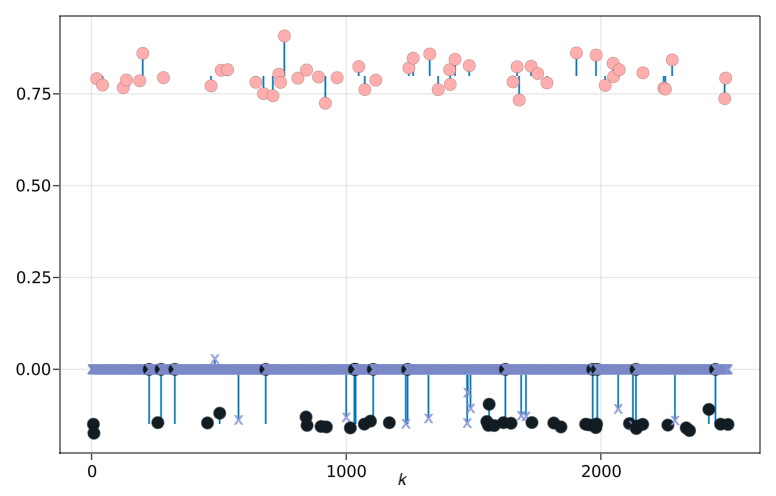}
    \caption{Posterior median estimate of $x_k$ (where $k$ can be identified with $(i,j)$) versus $k$ computed using the Sticky Zig-Zag sampler.  Thin vertical lines indicate distance to the truth. True zeros are plotted with the symbol $\times$, others are plotted as points. With $p = 50$ agents, the dimension of the problem is $d = 2450$ }
    \label{fig:boids}
\end{figure}
\begin{figure}[ht!]
    \centering
    \includegraphics[width=0.45\linewidth]{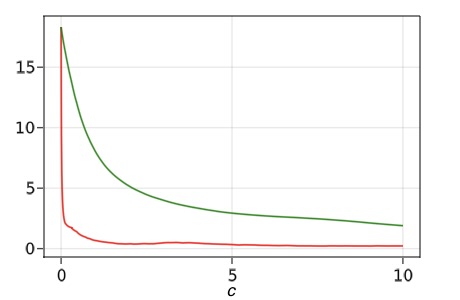}
    \caption{Squared error of the marginal inclusion probabilities (Equation~\ref{eq: squared error}) $c \to \cE_{\text{zig-zag}}(c)$ (red) and $c \to \cE_{\text{gibbs}}(c)$(green) where $c$ represent the computing time in seconds. With $p=20$ agents the dimension of the problem is $p(p-1)/2 = 380$.}
    \label{fig:boids_comparisons}
\end{figure}

\subsection{Spatially structured sparsity}\label{subsec: spatially structured sparsity}\label{sec: example>Spatially structured sparsity}
We consider the problem of denoising a  spatially correlated, sparse signal. The signal is assumed to be an $n\times n$-image. Denote the observed pixel value at location $(i,j)$  by $Y_{i,j}$ and assume 
\[
Y_{i,j} = x_{i,j} + Z_{i,j}, \quad Z_{i,j} \stackrel{\operatorname{i.i.d.}}{\sim}  \operatorname{N}(0,\sigma^2),\qquad  i, j \in \{1, \dots, n\}. 
\] 
The ``true signal'' is given by $x=\{x_{i,j}\}_{i,j}$ and this is the parameter we aim to infer, while assuming $\sigma^2$ to be known. We view $x$ as a vector in $\RR^d$, with $d=n^2$ but use both linear indexing $x_k$ and Cartesian indexing $x_{i,j}$ to refer to the component at index $k = n(i-1) + j$. The log-likelihood of the parameter $x$ is given by $\ell(x) = C + \sigma^{-2} \sum_{i=1}^n \sum_{j=1}^n |x_{i,j}-Y_{i,j}|^2$, with $C$ a constant not depending on $x$.

We consider the following prior measure 
\[ \mu_0(\dd x) = \exp\left(-\frac12 x'\Gamma x\right) \prod_{i=1}^d \left(\dd x_i + \frac{1}\kappa \delta_0(\dd x_i)\right).
\]
The Dirac masses in the prior encapsulate sparseness in the underlying signal and an appropriate choice of $\Gamma$ can promote smoothness. Overall, the prior encourages \emph{smoothness, sparsity and local clustering of zero entries and non-zero entries}. As a concrete example, consider 
$\Gamma = c_1\Lambda + c_2 I$  where $\Lambda$ is the graph Laplacian of the pixel neighbourhood graph:
the pixel indices $i, j$ are identified with the vertices $V=\{(i, j) \colon$ $(i,j)\in\{1, \ldots, n\}^2\}$ of the $n \times n$ -lattice with edges $E=\{\{v, v^{\prime}\}: (v, v^\prime) = ((i, j), (i^{\prime}, j^{\prime})) \in V^2$, $|i-i^{\prime}|+|j-j^{\prime}|= 1\}$ (using the set notation for edges). Thus, edges connect a pixel to its vertical and horizontal neighbours. Then
\begin{equation*}
\lambda_{v, v^{\prime}}=\left\{\begin{array}{ll}
\operatorname{degree}(v) & v=v^{\prime} \\
-1 & \left\{v, v^{\prime}\right\} \in E \\
0 & \text { otherwise }
\end{array}\right.
\end{equation*}
and 
$
\Lambda = (\Lambda_{k,l})_{k,l \in\{ 1, \dots, n^2\}}$
 with $\Lambda_{(i-1)n + j, (k-1)n + l} = \lambda_{(i,j), (k,l)}$,  for $\quad i,j,k,l \in \{1, \dots, n\}$.

This is a prior which is applicable in similar situations as the fused Lasso in \textcite{https://doi.org/10.1111/j.1467-9868.2005.00490.x}.
\medskip

\textbf{Numerical experiments:}
We assume that pixel $(i, j)$
corresponds to a physical location of size $\Delta_1 \times \Delta_2$ centered at $u(i,j) = u_0 + (i \Delta_1, j \Delta_2) \in \RR^2$. To numerically illustrate our approach,  we use a heart shaped region given by $
x_{i,j} =  5\max(1 - h(u(i,j)), 0)
$
where 
$h\colon \RR^2 \to [0, \infty)$  is defined by 
$
 h(u_1, u_2) = u_1^2+\left(\frac{5u_2}{4}-\sqrt{|u_1|}\right)^2 $, 
$u_0 = (-4.5, -4.1)$, $n = 10^3$ and $\Delta_1 = \Delta_2 = 9/n$.
In the example, about 97\% of the pixels of the truth are black. The dimension of the parameter equals  $10^6$.
Figure~\ref{fig:heartresults}, top-left, shows the observation $Y$ with $\sigma^2 = 0.5$  and the ground truth. 

As the ordinary Sticky Zig-Zag sampler would require storing and ordering $1$ million elements in the priority queue 
we ran the Sticky Zig-Zag sampler with sparse implementation as detailed in Remark~\ref{rmk: sparse implementation}. 
For this example, we have  $\Psi(x) = \ell(x)  +  0.5x'\Gamma x$. We took $c_1 = 2, c_2 = 0.1$ in the definition of $\Gamma$  and chose the  parameters $\kappa_1 = \kappa_2 = \dots = \kappa_d = 0.15$ for the smoothing prior. The reflection times are computed by means of a thinning scheme, see Appendix~\ref{app: spatially structured} for details. We set  the final clock of the Sticky Zig-Zag sampler to $500$.
Results from running the sampler are summarized in Figure~\ref{fig:heartresults}.

In Figure~\ref{fig: runtime of heart}, the runtimes of the Sticky Zig-Zag sampler and Gibbs sampler are shown (in a log-log scale)  for different values of $n^2$ (dimensionality of the problem), the  final clock was fixed to $T = 500$ ($10^3$  iteration for the Gibbs sampler). All the other parameters are kept fixed as described above.   The results agree well with the scaling results of Table~\ref{tab: 1}, rightmost column.

In Figure~\ref{fig: heart_comparison} we show $t \rightarrow \cE_{\text{Zig-Zag}}(t)$ and $t \rightarrow \cE_{\text{Gibbs}}(t)$ for $t$ ranging from 0 to $5$, in case  $n = 20$. Both samplers were initialized at the posterior mean of the Gaussian density proportional to $\Psi$ (hence, in the full-model with no coordinates set to 0). In this experiment, the Sticky Zig-Zag sampler outperforms the Gibbs sampler considerably.

\begin{figure}[ht!]
    \centering
        \centering
    \includegraphics[width=0.45\linewidth]{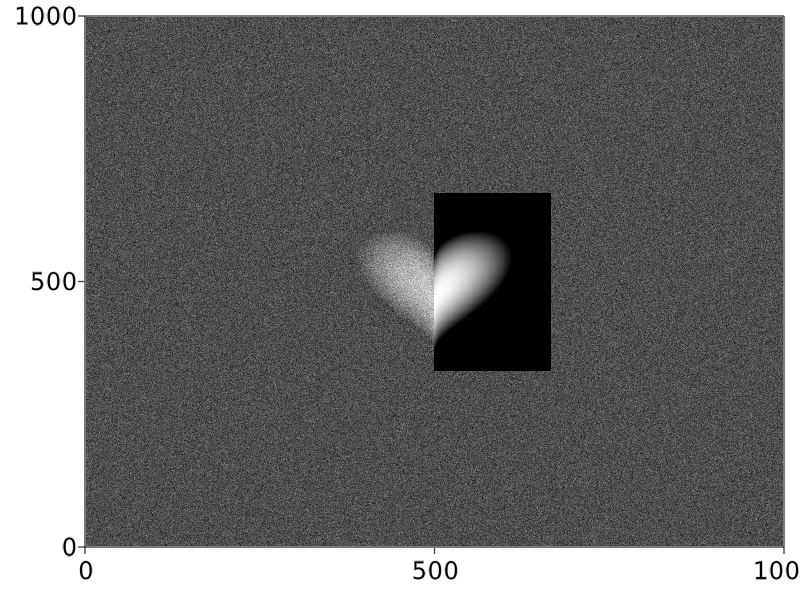}
    \includegraphics[width=0.45\linewidth]{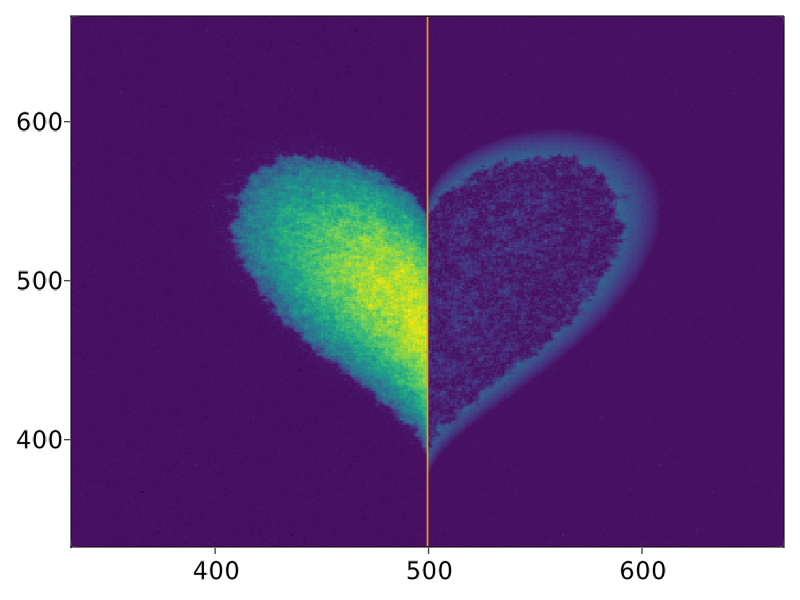}
    
    \includegraphics[width=0.45\linewidth]{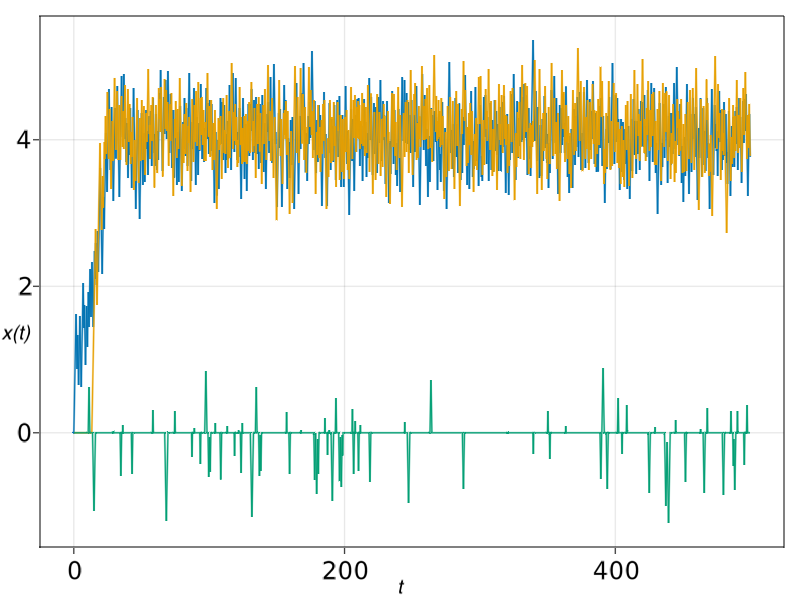}
    \includegraphics[width=0.45\linewidth]{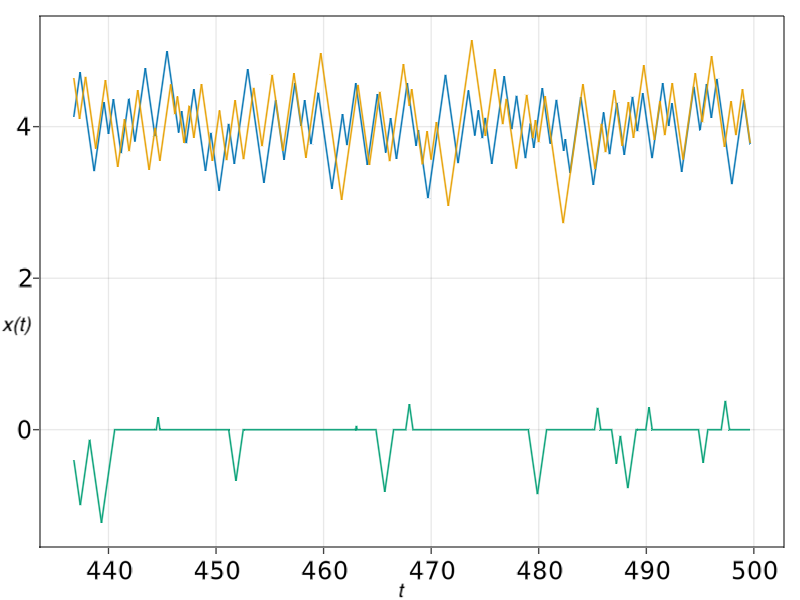}

    \caption{
    Top-left: observed $1\,000 \times  1\,000$ image of a heart corrupted with white noise, with part of the  ground truth inset. Top-right, left half: posterior mean estimated from the trace of the Sticky Zig-Zag sampler (detail). Top-right, right half: mirror image showing the absolute error between the posterior mean and the ground truth in the same scale (color gradient between blue (0) and yellow (maximum error)). 
    Bottom: trace plot of 3 coordinates; on the left the full trajectory is shown whereas on the right only the final $60$ time units are displayed. 
    The traces marked with blue and orange lines belong to neighbouring coordinates (highly correlated) from the center, the trace marked with green belongs to a coordinate outside the region of interest.}
    \label{fig:heartresults}
\end{figure}

\begin{figure}[ht!]
    \centering
        \includegraphics[width=0.5\linewidth]{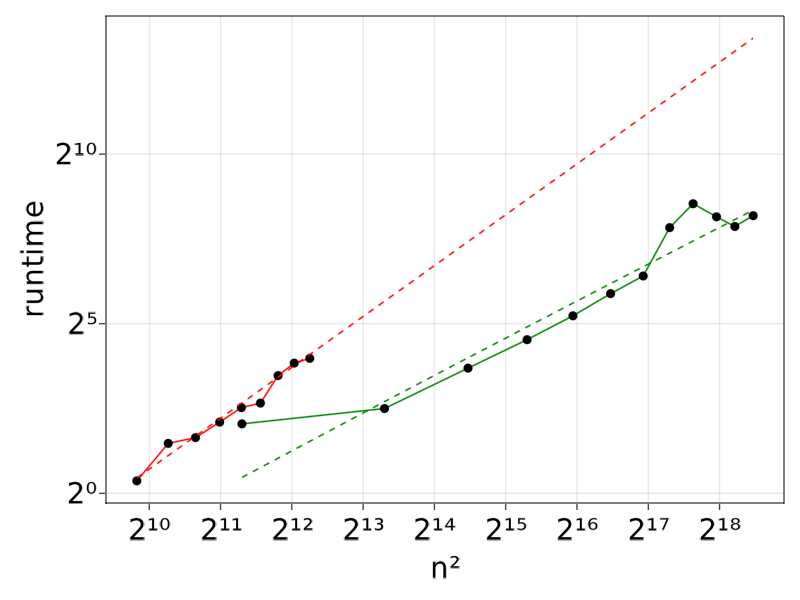}
    \caption{Runtime comparison  of the Sticky Zig-Zag sampler (green) and the Gibbs sampler (red) for the  example in Subsection~\ref{sec: example>Spatially structured sparsity}. The horizontal axis displays the dimension of the problem, which is $n^2$. The vertical axis shows runtime in seconds. The runtime is evaluated at $n^2 = 50^2, 100^2, \dots, 600^2$ for the sticky Zig-Zag sampler and at $n^2 = 40^2,45^2, \dots, 70^2$ for the Gibbs sampler.  Both plots are on a log-log scale. The dashed curves shows the theoretical scaling (including a log-factor for the priority queue insertion):  $x \mapsto c_1 x\log(x)$ (green) and $x \mapsto c_2  x^{3/2}$ (orange), with $c_1$ and $c_2$ chosen conveniently.}
    \label{fig: runtime of heart}
\end{figure}

\begin{figure}[ht!]
    \centering
    \includegraphics[width=0.45\linewidth]{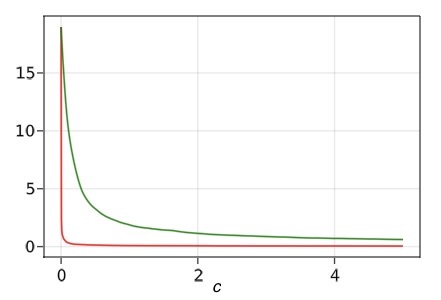}
    \caption{Squared error of the marginal inclusion probabilities (Equation~\ref{eq: squared error}) $c \to \cE_{\text{zig-zag}}(c)$ (red) and $t \to \cE_{\text{gibbs}}(c)$(green) where $c$ represent the computational time in seconds; right-panel: zoom-in near 0. Here the dimension of the problem is $n^2 = 400$.}
    \label{fig: heart_comparison}
\end{figure}

\subsection{Logistic regression}\label{subsec: examples>logistic regression}
Suppose $\{0,1\} \ni Y_i \mid x \sim \mbox{Ber}(\psi(x^T a_i))$ with $\psi(u) =(1 + e^{- u})^{-1}$. $a_i \in \RR^d$ denotes a vector of covariates and $x \in \RR^d$ a parameter vector. Assume $Y_1,\ldots, Y_N$ are independent, conditionally on $x$. The  log-likelihood is equal to
\[
 \ell(x)= \sum_{j=1}^N \left(\log \left(1 + e^{ \langle  a_{j}, x\rangle}\right) - y_j\langle a_{j}, x\rangle \right)
\]
We assume a spike-and-slab prior of \eqref{eq: spike-and-slab prior} with zeromean Gaussian slabs and  (large) variance $\sigma_0^2$. Then the posterior can be written as in Equation~\eqref{eq: target measure}, with $\Psi$ and $\kappa$ as in Equation~\eqref{eq: equivalence spike-and-slab}. 
\medskip

\textbf{Numerical experiments:}
We consider two categorical features with 30  levels each and 5 continuous features.
For each observation, an  independent random level of each discrete feature and a random value of the continuous features,  $\cN(0,0.1^2)$ is drawn. Let the design matrix $A\in \mathbb{R}^{N\times d}$ be the matrix where the $i$-th row is the vector $a_i$. $A$ includes the levels of the discrete features in dummy encoding and the interaction terms between them also in dummy encoding scaled by $0.3$ (960 columns), and the continuous features in the final 5 columns. This implies that the dimension of the parameter equals $d = 965$. We then generate $N =50d= 48250$ observations using as ground truth  sparse coefficients obtained by setting  $x_i = z_i \xi_i$ where $z_i \stackrel{\operatorname{i.i.d.}}{\sim} \text{Bern}(0.1)$ and $\xi_i\stackrel{\operatorname{i.i.d.}}{\sim} \cN(0, 5^2)$, where $\{z_i\}$ and $\{\xi_i\}$ are independent.

 We run the sticky ZigZag with subsampling and bounding rates derived in Appendix~\ref{app: logistic regression}. We chose $w_1 = w_2 = \dots = w_d = 0.1$ and $\sigma_0^2 = 10^2$ and ran the Sticky Zig-Zag sampler for $100$ time-units. The implementation makes use of a sparse matrix representation of $A$, speeding up the computation of inner products $\langle  a_{j}, x\rangle$.  Figure~\ref{fig:logistic_results} reveals that while perfect recovery is not obtained (as was to be expected), most nonzero/zero features \textit{are} recovered correctly. 
 
 In a second numerical experiment we compare the computing time of the Sticky Zig-Zag sampler and Gibbs sampler (as proposed in \cite{polson2013bayesian})  as we vary the number of observations ($N$). In this case, we reduce the dimension of the parameter by restricting to $2$ categorical variables, including their pairwise interactions, augmented by 3 ``continuous'' predictors (leading to the parameter vector $x \in \RR^9$). 
 For each sample size $N$ we ran the Gibbs sampler for $1000$ iterations and the  Sticky Zig-Zag sampler  for $1000$ time units. Our interest here is not to compare the computing time of the samplers for a fixed value of $N$, but rather the scaling of each algorithm with $N$.   Figure~\ref{fig:logistic_comparison} shows that the computing time for the Sticky Zig-Zag sampler is roughly constant when varying $N$. On the contrary, the computing time increases linearly  with $N$ for the Gibbs sampler. This is consistent with  the  theoretical scaling results presented in Table~\ref{tab: 1} (rightmost column). We remark that qualitatively similar results would be obtained if we would have fixed the number of iterations of the Gibbs sampler and endtime of the Zig-Zag sampler to different values. 
 
\begin{figure}[ht!]
    \begin{minipage}[t]{0.55\linewidth}
    \centering
    \includegraphics[width=\linewidth]{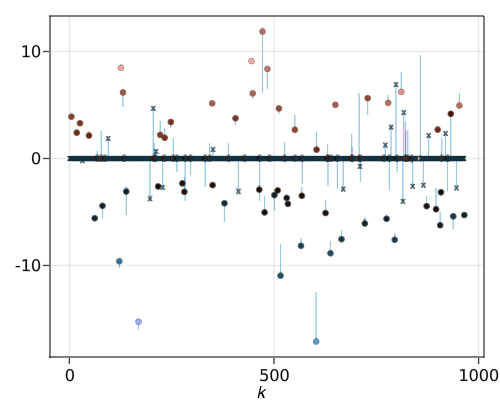}
    \caption{Results for the logistic regression  coefficients  derived with the Sticky Zig-Zag sampler with subsampling. Description as in caption of Figure~\ref{fig:boids}. The dimension of this problem is $d = 965$.}
    \label{fig:logistic_results}
\end{minipage}%
    \hfill%
\begin{minipage}[t]{0.4\linewidth}
\centering
    \includegraphics[width=0.73\linewidth]{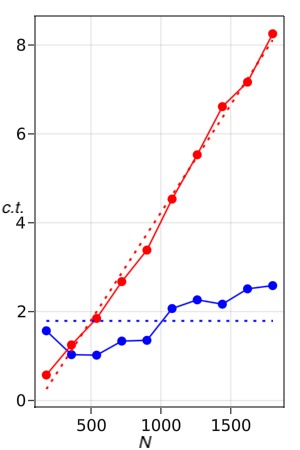}
        \caption{Logistic regression example: computing time in seconds versus number of observations. Solid red line: Gibbs samplers with $ 10^3$ iterations. Solid blue line: Sticky Zig-Zag samplers with subsampling ran for $10^3$ time units.  The dashed lines correspond to the  scaling results displayed in  Table~\ref{tab: 1}. Here, the dimension of the problem is fixed to $d = 9$.}
    \label{fig:logistic_comparison}
\end{minipage} 
    
\end{figure}

\subsection{Estimating a sparse precision matrix}\label{subsec: sparse precision matrix}
Consider  \[Y_i \mid X \stackrel{\operatorname{i.i.d.}}{\sim}  \cN_p\left(0, (X X')^{-1}\right), \quad i=1,2,\dots,N\]
for some unknown lower triangular sparse matrix $X \in \RR^{p\times p}$.  We aim to infer the lower-triangular elements of $X$ which we concatenate to obtain the parameter vector $x :=\{X_{i,j} \colon 1\le j \le i \le p\} \in \RR^{p(p+1)/2}$. This class of problems is important as the precision matrix $X X'$  unveils the conditional independence structure of $Y$, see for example  \textcite{shi2021bayesian}, and reference therein, for details.

We impose a prior measure on $x$ of the product form $\mu_0(\dd x) = \bigotimes_{i=1}^p \bigotimes_{j=1}^i  \mu_{i,j}(\dd x_{i,j})$ where
\[
\mu_{i,j}(\dd x_{i,j}) = \begin{cases}
\pi_{i,j}(x_{i,j}) \ind_{(x_{i,j >0 })} \dd x_{i,j}& \quad i = j,\\
w  \pi_{i,j}(x_{i,j}) \dd x_{i,j} + (1-w)\delta_0(\dd x_{i,j}) & \quad i \ne j,\\ 
\end{cases}  
\]
and $\pi_{i,j}$ is the univariate Gaussian density with mean $c_{i,j} \in \RR$ and  variance $\sigma_{0}^2>0$.

This prior induces sparsity on  the lower-triangular off-diagonal elements of  $X$  while preserving strict positive definiteness of  $X X'$ (as the elements on the diagonal are restricted to be positive). 

The posterior in this example is of the form 
\[ \mu(\dd x) \propto  \exp(-\Psi(x))\Big(\bigotimes_{i=1}^p \bigotimes_{j=1}^{i-1} (\dd x_{i,j} + \frac{1}{\kappa_{i,j}} \delta_0(\dd x_{i,j}))\Big)\bigotimes_{k=1}^p \dd x_{k,k} \]
 with \[
\Psi(x) = \frac12 \sum_{i=1}^N Y_i' X X 'Y_i -N \sum_{i = 1}^p\log(x_{i,i}) + \sum_{i=1}^p \sum_{j=1}^{i-1} \frac{(x_{i,j}- c_{i,j})^2}{2\sigma^2_{0}} + \sum_{i=1}^p \frac{(x_{i,i}- c_{i,i})^2}{2\sigma^2_{0}}
\]
and $\kappa_{i,j}= \pi_{i,j}(0) w/(1-w)$.
In particular, the posterior density is not of the form as 
given in  Equation~\eqref{eq: target measure}, as the diagonal elements cannot be zero and have a marginal density relative to the Lebesgue measure, while the off-diagonal elements are marginally mixtures of a Dirac and a continuous component. Notice that, for any $i = 1,2,\dots, p$, as $x_{i,i} \downarrow 0$,  $\exp(-\Psi(x))$ vanishes and $\nabla \Psi(x)\to \infty$. This makes the sampling problem challenging for gradient-based algorithms.

\medskip

\textbf{Numerical experiments:}
We apply the Sticky Zig-Zag sampler where the reflection times are computed by using a thinning and superposition scheme for inhomogeneous Poisson processes, see Appendix~\ref{app:sparse_precision_matrix} for the details.

We simulate realisations $y_1, \ldots, y_N$ with precision matrix $X X'$
 a tri-diagonal matrix with diagonal $(0.5, 1, 1,\dots, 1, 1, 0.5) \in \RR^p$ and  off-diagonal $(-0.3,-0.3,\dots,-0.3)\in \RR^{p-1}$. In the prior we chose $\sigma_{0}^2 = 10$ and $c_{i,j} = \ind_{(i=j)}$ and for $1 \le j \le i \le p$ and $w=0.2$.
   
  We fixed $N = 10^3$ and $p = 200$ and  ran the Sticky Zig-Zag sampler for $600$ time-units. We initialized the algorithm at $x(0) \sim \cN_{p(p+1)/2}(0,I)$ and set a burn-in of 10 unit-time. The left panel of Figure~\ref{fig:precision}   shows the error between $X X'$ (the ground truth) and $\overline X \, \overline X'$ where $\overline X$ is posterior mean of the lower triangular matrix estimated with the sampler. The error is concentrated on the non-zero elements of the matrix while the zero elements are estimated with essentially no error. The right panel of Figure~\ref{fig:precision}  shows the trajectories of two representative non-zero elements of $X$. The traces show qualitatively that the process converges quickly to its stationary measure.
  In this case, comparisons with the Gibbs sampler are not possible as there is no closed form expression for the Bayes factors of Equation~\eqref{eq; Bayes_factors}.
\begin{figure}[ht!]
\begin{center}
\includegraphics[width=0.95\linewidth]{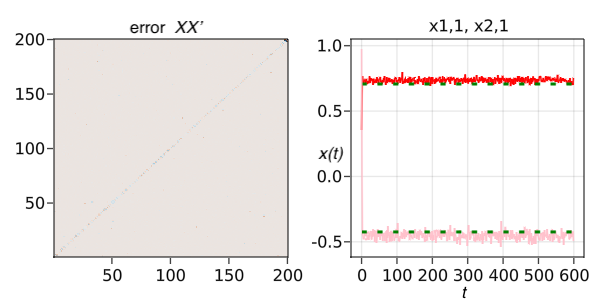}
\end{center}
\caption{Left: error between the true precision matrix and the precision matrix obtained with the estimated posterior mean of the lower-triangular matrix (colour gradient between white (no error) and black (maximum error)). Right: traces of two non-zero coefficients ($x_{1,1}$ in red and $x_{2,1}$ in pink) of the lower triangular matrix. Dashed green lines are the ground truth. Here, the dimension of each vector $Y_i$ is $p = 200$ and the dimension of the problem is $p(p+1)/2 = 20\,100$.}\label{fig:precision}
\end{figure}

\section{Discussion}\label{sec: discussions}
The sticky Zig-Zag sampler inherits some limitations from the ordinary Zig-Zag sampler: 

Firstly, if it is not possible to simulate the reflection times according to the Poisson rates in Equation~\eqref{eq: poisson rate zigzag}, the user needs to find and specify upper bounds of the Poisson rates from which it is possible to simulate the first event time (see Appendix~\ref{app: simulating sticky PDMPs} for details). This procedure is referred to as \emph{thinning} and remains the main challenge when simulating the Zig-Zag sampler. Furthermore, the efficiency of the algorithm deteriorates if the upper bounds are not tight.

Secondly, the Sticky Zig-Zag sampler, due to its continuous dynamics,  can experience difficulty traversing regions of low density, in particular it will have difficulty reaching 0 in a coordinate if that requires passing through such a region. 

Finally, the process can set to 0 (and not 0) only one coordinate at a time, hence failing to be ergodic for measures not supported on neighbouring sub-models.  For example, consider the space $\RR^2$ and assumes that the process can visit either the origin $(0,0)$ or the full space $\RR^2$ but not the coordinate axes $\{0\}\times\RR \cup \RR \times \{0\}$. Then the process started in $\RR^2$ hits the origin with probability $0$, hence failing to explore the subspace $(0,0)$.


In what follows, we outline promising research directions deferred to future work.  
\subsection{Sticky Hamiltonian Monte Carlo}
The ordinary Hamiltonian Monte Carlo (HMC) process as presented by \textcite{neal2011mcmc} can be seen as a piecewise deterministic Markov processes with deterministic dynamics equal to
\begin{equation}\label{eq: hamiltonian dynamics}
\dot x =  v, \qquad \dot v = - \nabla \Psi(x)  
\end{equation}
where $\nabla \Psi$ is the gradient of the negated log-density relative to the Lebesgue measure. At random exponential times with constant rate, the velocity component is refreshed as $v \sim \cN(0,I)$ (similarly to the refreshment events in the bouncy particle sampler). By applying the same principles outlined in Section~\ref{sec: sticky PDMP samplers}, such process can be made sticky with Equation~\eqref{eq: target measure} as its stationary measure.
 
Unfortunately, in most cases, the dynamics in \eqref{eq: hamiltonian dynamics} cannot be integrated analytically so that a sophisticated numerical integrator is usually employed and a Metropolis-Hasting steps compensates for the bias of the numerical integrator (see \cite{neal2011mcmc} for details). These two last steps makes the process effectively a discrete-time process and its generalization with sticky dynamics is not anymore trivial. 
\subsection{Extensions}
The setting considered in this work does not incorporate some relevant classes of measures:
\begin{itemize}
    \item Posteriors given by prior measures which freely choose prior weights for each (sub-)model. This limitation is mainly imputed to the parameter $\kappa =(\kappa_1,\kappa_2,\dots,\kappa_d)$ which here does not depend on the location component $x$ of the state space. While the theoretical framework built can be easily adapted for letting $\kappa$ depend on $x$, it is currently unclear to us the exact relationship between $\kappa$ and the posterior measure in this more general setting. 
    \item  Measures which are not supported on neighbouring sub-models are also not covered here. 
    To solve this problem, different dynamics for the process should be developed which allow the process to jump in space and set multiple coordinates to 0 (and not 0) at a time. 
\end{itemize}


\textbf{Acknowledgement:}
this work is part of the research programme \textit{Bayesian inference for high dimensional processes} with project number 613.009.034c, which is (partly) financed by the Dutch Research Council (NWO) under the  \textit{Stochastics -- Theoretical and Applied Research} (STAR) grant. J.~Bierkens acknowledges support by the NWO for the research project \textit{Zig-zagging through computational barriers} with project number 016.Vidi.189.043. 

\printbibliography 

\appendix

\clearpage 

\section{Details of the Sticky Zig-Zag sampler}\label{app: details of the sticky zz sampler}
\subsection{Construction}\label{app : construction sticky zz}
\newcommand{\stimes}{\!\times\!}
In this section we discuss how the Sticky Zig-Zag can be constructed as a \emph{standard} PDMP in the sense of \textcite{Davis1993}.
The construction  is a bit tedious, but the underlying idea is simple: the Sticky Zig-Zag process has the dynamics of a ordinary Zig-Zag process until it reaches a freezing boundary $\mathfrak F_i =\{(x,v) \in E\colon x_i = 0^-, \, v_i >0  \, \text{ or } \, x_i = 0^+, \, v_i < 0\}$ of $E = \overline\RR^d \times \cV$, with $\overline \R = (-\infty,0^-] \sqcup [0^+,\infty)$ which has two copies of $0$.  Then it immediately changes dynamics and evolves as a lower dimensional ordinary Zig-Zag process \emph{on the boundary}, at least until an unfreezing event happens or upon reaching yet another freezing boundary in the domain of the restricted process. 

Davis' construction allows a standard PDMP  to make instantaneous jumps at boundaries of open sets, but puts restrictions on further behaviour at that boundary. 
We circumvent these restrictions by first splitting up the space  $\RR^d \times \cV$ into disconnected components in a way somewhat different than the construction of $E$ as presented in Section~\ref{sec: sticky PDMP samplers}. Only at a later stage we recover the definition of $E$.

Define the set
\[ K = \{ \belowtozero ,  \aboveaway , \belowaway, \abovetozero, \zerofromabove, \zerofrombelow \} \]
and 
\[
|K| = \{\atzero, \awayfromzero, \towardzero\}
\]
(note that $|K|$ does not denote the cardinality of the set $K$). 
Define the functions $k \colon \RR \times \RR  \to K$ and $|k| \colon \RR \times \RR \to |K|$ by
\begin{center}
    \begin{tabular}[ht!]{|c|c|c|c|}
         \hline $(x,v)$ & $k(x,v)$ & at $(x,v)$  the process is... & $|k|(x,v)$ \\ 
         \hline \hline $x>0, v>0$ & $\aboveaway$ & ...moving away from 0 with positive velocity & $\awayfromzero$\\
         \hline $x<0, v<0$ & $\belowaway$ & ...moving away from 0 with negative velocity  & $\awayfromzero$ \\
         \hline $x>0, v<0$ & $\abovetozero$ & ...moving toward 0 with negative velocity  & $\towardzero$\\
         \hline $x<0, v>0$ & $\belowtozero$ & ...moving toward 0 with positive velocity  & $\towardzero$\\
         \hline $x = 0, v>0$ & $\zerofrombelow$ & ...at 0 with positive velocity & $\atzero$\\
         \hline $x = 0, v<0$ & $\zerofromabove$ & ...at 0 with negative velocity  & $\atzero$ \\
         \hline
    \end{tabular}
\end{center}

If $(x,v) \in \RR^d \times \mathcal{V}$, then extend $k \colon \overline\RR^d \times \mathcal{V} \to K^d$ and $|k| \colon \overline\RR^d \times \mathcal{V}\to |K|^d$ by applying the map $k$ and $|k|$ coordinatewise.

For each $\ell \in K^d$ define
\[ \tilde{E}_\ell^\circ =\{ (\ell, x, v) \colon k(x,v) = \ell \} \] 

Note that for $\ell \neq \ell'$ the sets $\tilde{E}_\ell^\circ$ and $\tilde{E}_{\ell'}^\circ$ are disjoint. 
The set $\tilde{E}_\ell^\circ$
 is open under the metric introduced in \textcite{Davis1993}, p.58, which sets the distance between two points $(\ell, x,v)$ and $(\ell',x',v')$ to 1 if $\ell \ne \ell'$. We denote the induced topology on $\tilde E$ by $\tilde \tau$.
 $\tilde E^\circ_\ell$ is a subset of $\RR^{2d}$ of dimension $d_\ell = \sum_{i =1}^d \1_{|\ell_i| \ne \atzero}$, since the velocities are constant in $E^\circ_\ell$ and the position of the components $i$ where $\ell_i = \atzero $ are constant as well in $\tilde E^\circ_\ell$ ($\tilde E^\circ_\ell$ is isomorphic to an open subset of $\RR^{d_\ell}$).
 
The sets which contain a singleton, i.e.\ $|\tilde E^\circ_\ell| = 1$,  are those sets $\tilde E^\circ_\ell$ such that $|\ell_i(x,v)| = \atzero$ for all  $i = 1,2,\dots, d$ and are open as they contain one isolated point, but will have to be treated a bit differently. Then $\tilde E^\circ = \bigcup_{\ell \in K^{d}} \tilde E^\circ_\ell$ is the tagged space of open subsets of $\RR^{d_\ell}$ used in \textcite[Section 24]{Davis1993}.

$\tilde E^\circ$ separates the space into isolated components of varying dimension. In each component, the Sticky Zig-Zag process behaves differently and essentially as a lower dimensional Zig-Zag process.

Let $\partial\tilde E_\ell^\circ$ denote the boundary of $\tilde E_\ell^\circ$ in the embedding space $\RR^{d_\ell}$ (where the velocity components are constant in $\tilde E_\ell^\circ$), with elements written $(\ell, x, v)$.
Some points in $\partial\tilde E_\ell^\circ$ will also belong to the state space $\tilde E$ of the Sticky Zig-Zag process, but only the entrance-non-exit boundary points:
\[
\tilde E = \bigcup_\ell \tilde E_\ell, \quad \tilde E_\ell = \tilde E_\ell^\circ \cup\{(\ell, x, v) \in \partial \tilde E^\circ_\ell \colon x_i = 0 \Rightarrow  |\ell_i| \ne \towardzero \text{ for all }i \}.
\]
(This corresponds to the definition of the state space in \cite[Section 24]{Davis1993}, only that we use knowledge of the flow.)

The remaining part of the boundary is
\[
 \Gamma = \bigcup_\ell \Gamma_\ell\subset \bigcup_\ell \partial \tilde E_\ell^\circ , \quad  \Gamma_\ell = \{(\ell, x, v) \in \partial \tilde E^\circ_\ell, \exists i \colon x_i = 0, |\ell_i|= \towardzero \},
\]
with $\tilde E \cap \Gamma = \varnothing$ so that $\Gamma$ is not part of the state space $\tilde E$. Any trajectory approaching $\Gamma$, jumps back into $\tilde E$ just before hitting $\Gamma$. 
If $\tilde E^\circ_\ell$ is a singleton $(|\tilde E^\circ_\ell| = 1)$, then $ \Gamma_\ell = \emptyset$ and  $\tilde E_\ell = \tilde E^\circ_\ell$ (atoms). 

\begin{lemma}
A bijection 
$\iota\colon \tilde E \to E$ is given by
\[
\iota((\ell, \tilde x, v)) =  (x,  v)
\]
where
\[x_i =
\begin{cases}
0^{+} \,(0^{-}) & \ell_i=\zerofromabove\, (\ell_i= \zerofrombelow)\\
0^{+} \,(0^{-}) & \ell_i=\aboveaway \,(\ell_i= \belowaway), \, \tilde x_i = 0 \\
\tilde x_i & \text{otherwise.}
\end{cases}
\]

\end{lemma}
\begin{proof}

Recall that $\alpha(x,v) := \{i\in \{1,2,\dots, d\} \colon (x, v) \notin  \mathfrak F_i \}$ and $\alpha^c$ denotes its complement. 
First of all, notice that $\iota(\tilde E) \subset E$. Now let $(x,v)\in E$ be given. We construct $e \in \tilde E$ such that $(x,v) = \iota(e)$. If there is at least one $x_j = 0^\pm$ with $j \notin \alpha(x,v)$, then take $e = (\ell, \tilde x,v) \in \tilde E\setminus \tilde E^\circ$
as follows (entrance-non-exit boundary): for $i \in \alpha^C$ we have $|\ell_i| = \atzero, \, \tilde x_i = 0$, while for all $i \in \alpha$ with $x_i = 0^\pm$, we have $|\ell_i| = \awayfromzero, \, \tilde x_i = 0$. Then $\iota(e) = (x, v)$. Otherwise, $e =(k(\tilde x, v), \tilde x,v)) \in \tilde E^\circ$ (interior of an open set) and $\iota(e) = (x, v)$ where $\tilde x_i = 0$ for all $i \in \alpha(x, v)$ and $\tilde x_i = x_i$ otherwise.

\end{proof}

Having constructed the state space, we proceed with the process dynamics. 
Firstly, the deterministic flow (locally Lipschitz for every $\ell \in K$) is determined by the functions $\tilde \phi_\ell\colon [0,\infty)\times \tilde E^\circ_\ell \to \tilde E^\circ_\ell$ which for the sticky ZigZag process are given by 
\[\tilde \phi(t, \ell, x, v) = (\ell, x', v), \quad \forall (\ell, x, v) \in E,
\]
with $x_i + v_i t (\1_{|\ell_i| \ne \atzero}), i= 1,2,\dots,d$
and determines the vector fields  
\begin{equation*}
    \mathfrak X_\ell \tilde f(\ell, x, v) = \sum_{i =1}^d  \1_{|\ell_i| \ne \atzero} v_i \partial_{x_i} f(\ell, x, v), \quad f \in C^1(\tilde E).
\end{equation*}
Sometimes we write $\tilde \phi_k(t, x, v) = \tilde \phi(t, k, x, v)$ for convenience.
Next, further state changes of the process are instantaneous, deterministic jumps from the boundary $\Gamma$ into $\tilde E$ 
\[
\cQ^{\rm f}(((\ell, x,v), \cdot)) =  \delta_{(k(x,v), x, v)}, \quad (\ell, x, v) \in \Gamma\\
\]
and random jumps at random times corresponding to unfreezing events 
\[
\cQ^{\rm s}((\ell, x,v), \cdot) =
\frac{\sum_i \lambda_i^{\rm s}(\ell,x,v) \delta_{(\ell[i\colon \ell_i'], x, v)}}{\sum_{i} \lambda_i^{\rm s}(i,x,v)}
\]
with $\ell_i' = \aboveaway$ if $\ell_i = \zerofrombelow$ and $\ell_i' = \belowaway$ if $\ell_i =\zerofromabove$,  and random reflections
\[\cQ^{\rm r}((\ell,x,v), \cdot) = \frac{\sum_{i} \lambda^{\rm r}_i(\ell,x,v) \delta_x  \delta_{v[i\colon -v_i]}\delta_{\ell}}{\sum_{i}\lambda_i^{\rm r}(\ell, x, v)} \]
with
\[
\lambda_i^{\rm s}(\ell, x, v)= \1_{|\ell_i| = \atzero} \kappa_i
\]
and
\[
\lambda^{\rm r}_i(\ell,x,v) = \1_{\ell_i \ne \atzero} \left((v_i \partial_i \Psi(x))^+ + \lambda_{0, i}(x)\right), \quad i=1,2,\dots,d.
\]
 
Then $\lambda  \colon \tilde E \to \RR^+$
\[
\lambda(\ell, x, v) =\sum_{i=1}^d\lambda_i^{\rm r}(\ell, x, v) +  \lambda_i^{\rm s}(i,x,v)
\]
and  a Markov kernel  $\cQ\colon (\tilde E \cup 
\Gamma , \cB(\tilde E \cup 
\Gamma )) \to [0,1]$ by 
\[
\cQ((\ell,x,v), .) = 
\begin{cases}
\frac{\sum_{i}\lambda_i^{\rm r}(\ell, x, v)}{\lambda(\ell, x, v)}\cQ^{\rm r}((\ell,x,v), .) + \frac{\sum_{i}\lambda_i^{\rm s}(\ell, x, v)}{\lambda(\ell, x, v)} \cQ^{\rm s}((\ell,x,v), .) & (\ell, x, v) \in \tilde E,\\
\cQ^{\rm f}((\ell,x,v), .) & (\ell,x,v) \in \Gamma.
\end{cases}
\]
\begin{proposition}
  
 $\mathfrak X, \lambda, \cQ$ satisfy the \emph{standard conditions} given in \textcite[Section 24.8]{Davis1993}, namely \begin{itemize}
     \item For each $\ell \in K,\, \mathfrak X_\ell$ is a locally Lipschitz continuous vector field and determines the deterministic flow $\tilde\phi_\ell:\tilde E_\ell \to \tilde E_\ell $ of the PDMP.
    \item $\lambda \colon \tilde E \to \RR^+$ is measurable and such that $t \to \lambda(\tilde\phi_\ell(t, x, v))$ is integrable on $[ 0, \epsilon(\ell,x,v))$, for some $\epsilon>0$, for each $\ell,x,v$.  
    \item $\cQ$ is measurable and such that $\cQ((\ell, x, v), \{(\ell, x, v)\}) = 0$
    \item The expected number of events up to time $t$, starting at $(\ell,x,v)$ is finite for each $t>0, \forall (\ell,x,v) \in \tilde E$
     \end{itemize} 
\end{proposition}
To see the latter, remember that for any initial point $(\ell, x, v) \in \tilde E$, the deterministic flow (without any random event) hits $\Gamma$ at most $d$ times before reaching the singleton $(0,0,\dots,0)$ and being constant there. 

\subsection{Strong Markov property}\label{app: strong markov property sticky zz}
\begin{proposition}  (Part of Theorem~\ref{thm: Sticky Zig-Zag main result})
Let $(\tilde Z_t)$ be a Zig-Zag process on $\tilde E$ with characteristics $\mathfrak X, \lambda, \cQ$. Then $Z_t = \iota(\tilde Z_t)$ is a strong Markov process. 
\end{proposition}
\begin{proof}
By \textcite{Davis1993}, Theorem 26.14, the domain of the extended generator of the process $(\tilde Z_t)$ with characteristics $\mathfrak X, \lambda, \cQ$ is  
\begin{align*}
    \cD(\tilde \cA) = \{f \in \cM(\tilde E); \, &t \to f(\tilde\phi_\ell(t,x, v)) \text{ $\tilde \tau$-absolutely continuous } \forall (\ell,x,v) \in \tilde E, t = [0, t_\Gamma(\ell,x,v)); \\
    &f(\ell, x, v) = f(\kappa(x,v),x,v),\quad (\ell, x, v)\in \Gamma\},
\end{align*}
with 
\[
t_{\Gamma}(\ell, x, v) =  \inf\{0\le t \colon  \tilde\phi_\ell(t, x, v) \in \tilde \Gamma\}
\]
and
\[
\tilde \cA f(\ell,x,v) = \mathfrak X_\ell f(\ell,x,v) + \lambda(\ell,x,v) \int_{\tilde E} (f(\ell',x',v') - f(\ell,x,v)) Q(\ell,x,v, \dd(\ell,x,v)). 
\]

The strong Markov property of $(\tilde Z_t)$ follows by \textcite{Davis1993}, Theorem 25.5. Denote by $(\tilde P_t)_{t\ge0}$ the Markov transition semigroup of $(\tilde Z_t)$ and let $(P_t)_{t\ge0}$ be a family of probability kernels on $E$ and such that for any bounded measurable function  $f \colon E \to \RR$ and any $t\ge0$,
\[
\tilde P_t (f \circ \iota) = (P_t f)\circ \iota.
\]
Then $(P_t)_{t\ge0}$ is the Markov transition semigroup of the process $Z_t = (\iota(\tilde Z_t))$. By \textcite{rogers2000diffusions}, Lemma 14.1, and since any stopping time for the filtration of $(\tilde Z_t)$ is a stopping time for the filtration of $(Z_t)$, $Z_t$ is  a \emph{strong} Markov process.

\end{proof}

\subsection{Feller property}\label{app: feller property sticky zz}
Given an initial point $\ell, x,v \in \tilde E$, let  
\[
t_{\Gamma_1}(\ell, x, v) =  \inf\{0\le t \colon \tilde \phi_\ell(t, x, v) \in \tilde \Gamma\}
\]
and define the \emph{extended deterministic flow} $\tilde \varphi\colon \tilde E \to \tilde E$ by setting $\varphi(0, \ell, x, v) = (\ell, x, v)$ and recursively by
\[
\tilde \varphi(t, \ell, x , v) = \begin{cases}  \tilde \varphi_\ell(t, x, v) & t < t_{\Gamma_1},  \\
\tilde  \varphi (t - t_{\Gamma_1}, k(x',v'), x', v') \, & t \ge  t_{\Gamma_1} 
\end{cases} 
\]
with $(\ell', x',v') = \lim_{t \to t_{\Gamma_1}}\tilde \varphi_\ell(t, x, v) \in \Gamma$. 

Observe that $t \to \iota(\tilde \varphi(t, \ell,x,v))$ is continuous on $(E, \tau)$. 
Define also
\[
\Lambda(t, \ell, x,v) = \int_0^t \lambda(\tilde \varphi(s, \ell,x,v)) \dd s.
\]
Notice that, while $(\ell, x,v) \to \lambda(\ell,x,v)$ has discontinuities at the boundaries $\Gamma$, $(\ell, x,v) \to \Lambda(\ell,x,v)$ is continuous. Denote by $T_1$ the first random event (so excluding the deterministic jumps). Then for functions $f \in B(\tilde E)$ and $\psi \in B(\RR^+ \times \tilde E)$, set $z(t) = (\ell(t), x(t), v(t))$ and define
\[
\tilde G\psi(t,\ell,x,v) = E[f(z(t))\1_{t< T_1} + \psi(t- T_1, z(t))\1_{t\ge T_1}].
\]
We have that
\begin{equation} \label{eq: operator G}
           \tilde G\psi(t, \ell, x, v) = f(\tilde \varphi(t,\ell, x, v ))  \times  \mathcal{T} 
           \end{equation}
          with 
          \[ \mathcal{T} =       \sum_i \int_0^{t} \ind_{t \in [ t^\Gamma_i, t^\Gamma_{i+1})} \int_{x',v'} \psi(t - s, \ell, x, v)\cQ((\ell, \dd x', \dd v'), \tilde \varphi(s, \ell,x,v))\lambda(\tilde \varphi(s,\ell,x,v))e^{-\Lambda(s,\ell,x,v)} \dd s.\]
The Feller property holds if, for each fixed $t$ and for $f \in C_b( E)$, we have that $(x,v) \to P_t f(x,v)$ is continuous (and bounded follows easily). This is what we are going to prove below, by making a detour in the space $\tilde E$, using the bijection $\iota$ and adapting some results found in \textcite[Section 27]{Davis1993}, for the process $\tilde Z_t$.
\begin{theorem}
 (Part of Theorem~\ref{thm: Sticky Zig-Zag main result}) $Z_t$ is a Feller process. 
\end{theorem}
\begin{proof}
Take $f \in C_b(\tilde E)$ such that $f\circ \iota \in C_b(E)$.
Call those functions on $\tilde E$ $\tau$-continuous.
We want to show that $\tilde P$ preserves $\tau$-continuity.
Notice that $\tau$-continuous functions on $\tilde E$ are such that
\[\lim_{t \to t_\Gamma} f(\tilde \varphi(t, \ell,x,v)) =  f(\tilde \varphi(t_\Gamma, \ell,x,v))), \quad (\ell,x,v) \in \tilde E.\]
For $\tau$-continuous functions $f$ and for a fixed $t$, the first term on the right hand side of \eqref{eq: operator G} $(\ell,x,v) \to f(\tilde \varphi(t,\ell,x,v))$ is clearly continuous. Also the second term is continuous since is of the form of an integral of a piecewise continuous function. Therefore, for any $t\ge0, \, \psi(t, \cdot) \in B(\tilde E)$ and $\tau$-continuous function $f$, we have that $(\ell,x,v)\to \tilde G \psi(t,\ell,x,v)$ is continuous. Clearly, the (similar) operator 
\[\tilde G_n \psi_\ell(t, x, v) = E_x[f(\tilde \varphi_\ell(t,x,v))\1_{t< T_n} + \psi(t- T_n, \tilde \varphi_\ell(t, x,v))\1_{t\ge T_n}],\]
with $T_n$ denoting the $n$th random time, is continuous as well for any fixed $n,\, t,\, \psi(t, \cdot) \in B(\tilde E)$ and $\tau$-continuous function $f$. 
By applying Lemma 27.3 in \textcite{Davis1993} we have that for any $\psi(t, \cdot) \in B(\tilde E)$  
\[|\tilde G_n \psi_\ell(t, x, v) - \tilde P_t f(x, v)| \le 2\max(\|\psi\|\|f\|) P(t\ge T_n).
\]
Finally, if $\lambda$ is bounded, then we can bound  $P(t\ge T_n)$ by something which does not depend on $(\ell, x,v)$ and goes to 0 as $n \to \infty$ so that $\tilde G_n \psi \to \tilde P_t f$  uniformly on $\ell,x,v \in \tilde E$ under the supremum norm. This shows that, for any $t$, $\tilde P_t$ (and therefore $P_t$) preserves $\tau$-continuity.
\end{proof}
\begin{remark}
The proof of the Feller and Markov property follow similarly for the Bouncy Particle and the Boomerang sampler.  
\end{remark}

\subsection{The extended generator of \texorpdfstring{$Z_t$}{Zt}}\label{app: the extended generator of sticky zig zag}

  Let $f \in \cD(\cA)$ if $\tilde f \in \cD(\tilde \cA)$ and $ f \circ \iota = \tilde f$. Then $f\in \cD(\cA)$ are $\tau$-absolutely continuous functions along full deterministic trajectories on $E$:
\begin{align*}
    \cD(\cA) = \{f \in \cM(E); \, &t \to f(\varphi(t,x, v)) \text{ $\tau$-absolutely continuous } \forall (x,v);  \\
    & \lim_{t \to 0} f(x[i\colon 0^+ + t], v) = f(x[i\colon 0^+], v);\\
    & \lim_{t \to 0} f(x[i\colon 0^- - t], v) = f(x[i\colon 0^-], v)\}.
\end{align*}
For those functions $f \in \cD(\cA)$ with $ f \circ \iota = \tilde f$ we have that \[\tilde \cA \tilde f(\ell,\tilde x, v) = \cA f(x,v) =  \sum_{i=1}^N \cA_i f(x,v)\]
with
\[
\cA_i f(x,v) = \begin{cases} \kappa_i (f(T_i(x,v)) - f(x,v)) & (x,v) \in \mathfrak F_i, \\
v_i \partial_{x_i} f(x,v) + \lambda_i(x,v)(f(x, v[i\colon -v_i]) - f(x,v)), & \text{otherwise,}
\end{cases}
\]
and
\[
\lambda_i(x,v) = (v_i \partial_i \Psi(x))^+ + \lambda_{0,i}(x), \quad i=1,2,\dots,d,
\]
for positive functions $\lambda_{0,i}$.

Denote  the space of compactly supported functions on $E$ which are continuously differentiable in their first argument by $C^1_c(E)$. Define $C_b(E) = \{f \in C(E) \colon f \text{ is bounded}\}$ and
$ D = \{f \in C^1_c( E), \cA f \in C_b( E)\}
$. The following proposition shows that the operator $\cA$ restricted to $D$ coincides with the infinitesimal generator of the ordinary Zig-Zag process restricted to $D$.
\begin{proposition}\label{prop: ext generator restricted resable standard zigzag generator}
We have
\begin{multline*}
D = \{f \in C_c^1( E) \colon v_i \kappa_i \left( f(T_i(x, v)) - f(x, v)\right) \\ = v_i \partial_i f(x,v) + \lambda_i (x,v)(f(x, v[i\colon-v_i])) - f(x, v)), (x, v) \in \mathfrak{F}_i \ \text{for all $i =1,\dots, d$}\}.
\end{multline*}
 For $f \in D$, $\cA f = \cL f$, where $\cL f = \sum_{i=1}^d \cL_i f$ with
\[\Ge_i f(x, v) = v_i \partial_{x_i} f(x, v) + \lambda_i(x,v)\left(f(x, v[i\colon-v_i]) - f(x, v)\right).
\]  

\end{proposition}

\begin{proposition}\label{prop: ext generator, domain}
   (Proposition~\ref{prop: gen and ext gen of sticky zz}) The extended generator of the process $(Z(t))$ is given by $\mathcal A$ with domain $\mathcal D(\mathcal A)$.
\end{proposition}
\begin{proof}
    This is to verify that if $f \in \cD(\tilde \cA)$ and $\tilde \cA$ solve the martingale problem, i.e are such that 
     \[
      f(\ell(t), x(t), v(t)) - f(\ell, x, v) + \int_0^t \cA f(\ell(s), x(s), v(s) \dd s, \quad \forall (\ell,x,v)\in \tilde E 
      \]
    is a local martingale (\cite[Section 24]{Davis1993}) on $\tilde E$, then $f \circ \iota \colon f \in \cD(\tilde \cA)$ and $\cA$ solve the martingale problem on $E$ (for any local martingale $Z_t$ on $\tilde E$, $\iota(Z_t)$ is a local martingale on $E$).
\end{proof}
By the Feller property, the extended generator is an extension of the generator defined as 
\[
\mathcal{L}f(x,v) :=  \lim_{t\downarrow 0}\frac{\E[f(X_t,V_t) \mid X_0 = x, V_0 = v] - f(x, v)}{t}
\]
for a sufficient regular class of functions $f$ for which this limit exists uniformly in $x$ (see \cite[Section 3]{Liggett2010}, for more details).  Then, $D = \{f \in \cD(\cA)\colon f \in C_b^1, \,\cA f \in C_b(E)\}$ is a core for $\cA$ (as in \cite[Definition 3.31]{Liggett2010}). Let $\cL$ be the restriction of $\cA$ on $D$. By \textcite[Theorem 3.37]{Liggett2010}, $\mu$ is a stationary measure if, for all $f \in D$: 
\[
\int \cL f \dd \mu  = 0.
\]
\subsection{Remaining part of the proof}\label{app: remaining proofs of section 2}
\textbf{Invariant measure of the Sticky Zig-Zag process:}
We check here that the sticky $d$-dimensional Zig-Zag process as presented in Section~\ref{sec: sticky PDMP samplers>Sticky Zig-Zag sampler} taking values in $E$ with discrete velocities in $\mathcal{V} = \{v \colon |v_i| = a_i, \forall i \in \{1,2,\dots, d\}\}$ and with extended generator $\cA$ is such that 
\[
\int \Ge f(x, v)  \mu( \dd x, \dd v) = 0
\]
for all $f \in D = \{f \in C^1_c( E), \cA f \in C_b( E)\}$. Here, $\cL$ is the extended generator $\cA$ restricted to $D$ (See Poposition~\eqref{prop: ext generator restricted resable standard zigzag generator}).
For any-1 $f \in D$,
define $\lambda^+_i := \lambda_i(x, v[i\colon, a_i]), \,\lambda^-_i := \lambda_i(x, v[i\colon, -a_i]),\, f^+_i := f(x, v[i\colon a_i]),\, f^-_i := f(x, v[i\colon -a_i]),f^+_i(y) := f(x[i\colon y], v[i\colon a_i]),\, f^-_i(y) := f(x[i\colon y], v[i\colon -a_i]), .$ Also write the measure $\rho(\dd x_i, v_i) := \dd x_i + \frac1\kappa\left( \1_{v_i<0}\delta_0^+(\dd x_i) + \1_{v_i>0}\delta_0^-(\dd x_i)\right)$. We see that 

\begin{align*}
 \int \Ge_i &f \dd \mu =  \\
 &\sum_{v \in \mathcal{V}^{-i}} \left(\int_{\RR^{d-1}} \left(\int^{\infty}_{0^+} + \int_{-\infty}^{0^-}\right)\left( a_i \partial_{x_i} f^+_i + \lambda_i^+(f_i^- -  f_i^+)\right)\exp(-\Psi(x))\dd x_i  \prod_{j\ne i}\rho(\dd x_j, v_j )\right)\\ 
&+ \sum_{v \in \mathcal{V}^{-i}} \left(\int_{\RR^{d-1}} \left(\int^{\infty}_{0^+} + \int_{-\infty}^{0^-}\right)\left( -a_i \partial_{x_i} f^-_i + \lambda_i^-(f_i^+ - f_i^-)\right)\exp(-\Psi(x))\dd x_i  \prod_{j\ne i}\rho(\dd x_j, v_j ) \right)\\ 
&+ \sum_{v \in \mathcal{V}^{-i}}\left( \int_{\RR^{d-1}} a_i \left(f_i^+(0^+) - f_i^+(0^-)\right)\exp(-\Psi(x[i\colon 0])) \prod_{j\ne i}\rho(\dd x_j, v_j ) \right) \\
&+\sum_{v \in \mathcal{V}^{-i}}\left( \int_{\RR^{d-1}} -a_i \left(f_i^- (0^-) - f_i^-(0^+)\right)\exp(-\Psi(x[i\colon 0])) \prod_{j\ne i}\rho(\dd x_j, v_j )\right).
\end{align*}
By integration by parts  we have that $\left(\int^{\infty}_{0^+} + \int_{-\infty}^{0^-}\right)\left( \partial_{x_i} f(x, v)\exp(-\Psi(x)) \right)\dd x_i$ is equal to 
\[
  \left(f(x[i\colon 0^-],v) - f(x[i\colon 0^+],v)\right)\exp(-\Psi (x[i\colon 0])) + \left(\int^{\infty}_{0^+} + \int_{-\infty}^{0^-}\right)\left( \partial_i \Psi(x) f(x, v) \exp(- \Psi(x))\right) \dd x_i 
\]
so that $\int \Ge_i f \dd \mu$ is equal to
\begin{align*} 
\sum_{v \in \mathcal{V}^{-i}} &\left(\int_{\RR^{d-1}} \left(\int^{\infty}_{0^+} + \int_{-\infty}^{0^-}\right)\left( a_i \partial_{x_i} \Psi(x) + \lambda_i^+ - \lambda_i^- \right)f_i^-\exp(-\Psi(x))\dd x_i  \prod_{j\ne i}\rho(\dd x_j, v_j ) \right)\\ 
&+\sum_{v \in \mathcal{V}^{-i}} \left(\int_{\RR^{d-1}} \left(\int^{\infty}_{0^+} + \int_{-\infty}^{0^-}\right)\left( -a_i \partial_{x_i} \Psi(x) + \lambda_i^- - \lambda_i^+ \right)f_i^+\exp(-\Psi(x))\dd x_i  \prod_{j\ne i}\rho(\dd x_j, v_j ) \right)\\ 
&+ \sum_{v \in \mathcal{V}^{-i}}\left( \int_{\RR^{d-1}} a_i \left(f_i^+(0^+) - f_i^+(0^-)\right)\exp(-\Psi(x[i\colon 0])) \prod_{j\ne i}\rho(\dd x_j, v_j )\right)\\
&+\sum_{v \in \mathcal{V}^{-i}}\left( \int_{\RR^{d-1}} -a_i \left(f_i^-(0^-) - f_i^-(0^+) \right)\exp(-\Psi(x[i\colon 0])) \prod_{j\ne i}\rho(\dd x_j, v_j )\right)\\
&+ \sum_{v \in \mathcal{V}^{-i}}\left( \int_{\RR^{d-1}} a_i \left(f_i^+(0^-) - f_i^+(0^+)\right)\exp(-\Psi(x[i\colon 0])) \prod_{j\ne i}\rho(\dd x_j, v_j )\right)\\
&+\sum_{v \in \mathcal{V}^{-i}}\left( \int_{\RR^{d-1}} -a_i \left(f_i^-(0^+) - f_i^-(0^-)\right)\exp(-\Psi(x[i\colon 0])) \prod_{j\ne i}\rho(\dd x_j, v_j ) \right) = 0,
\end{align*}
where we used that $-v_i\partial_i\Psi(x) + \lambda_i(x, v) - \lambda_i(x, F_i(v)) = 0, \, \forall (x, v) \in E$.

\subsection{Ergodicity of the sticky Zig-Zag process}\label{app: ergodicity of the szz}
In this section, we prove that the sticky Zig-Zag is ergodic. As the argument partially relies on the ergodicity results of the ordinary Zig-Zag sampler (\cite{bierkens2019ergodicity}), we start by making similar assumptions on $\Psi$ as appearing in that paper.
\begin{assumption}\label{ass: Assumption ergodicity}
    \emph{(Assumptions of \cite[Theorem~1)]{bierkens2019ergodicity}} Let $\Psi$ satisfy the following conditions: \begin{itemize}
        \item $\Psi \in \cC^3(\RR^d)$,
        \item $\Psi$ has a non degenerate local-minimum, 
        \item For some constants $c > d, \, c' \in \RR$, $\Psi(x) > c\ln(|x|) - c'$, for all $x \in \RR^d$.
    \end{itemize}
\end{assumption}
For every set $\alpha \subset \{1,2,\dots, d\}$, we define the sub-space $\cM_\alpha = \{x \in \RR^d \colon x_i = 0, \, i \notin \alpha\}$ and define the $|\alpha|$-dimensional ordinary Zig-Zag process $(Z_t^{(\alpha)})_{t\ge0}$, with $|\alpha| \le d$, on the sub-space $\cM_\alpha\times \{-1,+1\}^{\alpha}$ and with reflection rates $\lambda_i(x,v) = \max(0, v_i \partial_i \Psi(x))$, $x \in \mathcal M_{\alpha}$,  $i \in \alpha$.
\begin{proposition}\label{prop: ergodicity of subzigzags}Suppose $\Psi$ satisfies Assumption~\ref{ass: Assumption ergodicity}. Then for every set $\alpha \subset \{1,2,\dots, d\}$, $(Z^{(\alpha)}_t)_{t\ge 0}$ is ergodic with unique invariant measure with density $\left. \exp(-\Psi(x)) \right|_{\mathcal M_{\alpha}}$ relative to $\text{Leb}(\mathcal M_{\alpha})(\dd x) \otimes \text{Uniform}(\{-1,+1\}^{\alpha})(\dd v)$. Furthermore, some skeleton chain of each process is irreducible.
\end{proposition}
\begin{proof} If Assumption~\ref{ass: Assumption ergodicity} holds on $\RR^d$, then it holds on any the sub-space $\cM_\alpha, \, \alpha \subset \{1,2,\dots, d\}$, for functions $x \mapsto \Psi(x)$, $x \in \mathcal M_{\alpha}$. Proposition~\ref{prop: ergodicity of subzigzags} follows from the ergodic theorem of ordinary Zig-Zag processes (\cite[Theorem~1 and Theorem~5]{bierkens2019ergodicity}).
\end{proof}
Next, we show that, for any initial position $(x,v) \in E$, the sticky Zig-Zag process is Harris recurrent to the set where all coordinates are stuck at 0.
Denote the measure $\overline \delta_{0}(\dd x, \dd v) = \bigotimes_{i=1}^d (\delta_{0^+, -1}(\dd x_i, \dd v_i) + \delta_{0^-, +1}(\dd x_i, \dd v_i))$, the set $\fS = \cap_{i=1}^d \fF_i$ and the first hitting time $\tau_{A} = \inf \{t>0 \colon Z_t \in A\}$, where $Z_t = (X_t, V_t)$ is the sticky Zig-Zag process. 
\begin{proposition}\label{prop: harris recurrent}\emph{(Harris recurrence)}
Suppose $\Psi$ satisfies Assumption~\ref{ass: Assumption ergodicity}. Then for any initial state $Z_0 = z_0 \in E$, we have that ${\PP(\tau_\fS < \infty) = 1}$. 
\end{proposition}
\begin{proof}
Let $x_0 \in \cM_\alpha$ for an arbitrary $\alpha \subset \{1,2,\dots,d\}$. 
Denote the random time of the first stuck coordinate $x_i, \, i \in \alpha^c$ leaving zero by $T_1 \sim \text{Exp}(\sum_{j \in \alpha^c} \kappa_j)>0$. 
Denote the random time of the first `free' coordinate $x_i, \, i \in \alpha$ hitting zero by $T_2$. 


Notice that $T_1$ is independent of the trajectory on the subspace $\cM_\alpha$.
and the sticky Zig-Zag process behaves as an ordinary $|\alpha|$-dimensional Zig-Zag process in the subspace $\mathcal{M}_\alpha$ for time $t \in [0, \min(T_2, T_1)]$. 
By Proposition~\ref{prop: ergodicity of subzigzags}, $T_2$ is finite and $\PP(T_2 < T_1) > 0$.  
By using the Markov structure of the process and iterating the same argument for a sequence of sub-models $\cM_{\alpha_2},\, \cM_{\alpha_3},\, \dots, \cM_{\alpha_{|\alpha|-1}}$, with $|\alpha_j| + 1 =  |\alpha_{j+1}|$, we conclude that $P(\tau_{\fS} < \infty) =1$.

Now, consider a subset $S\subset \fS$, which contains points with some coordinates being either $(0^+, -1)$ or $(0^-, +1)$ (but not both). 
Let $\tau_{\fS}$ be the hitting time to the set $\fS$ of the sticky Zig-Zag  $Z(t)_{t>0}$. Denote by $\beta\subset\{1,2,\dots, d\}$ the set of indices for which the coordinate is stuck on the other copy of zero which is not in $S$. At time $Z(\tau_{\fS})$ the process will stay in the null model for a time $\Delta T \sim\text{Exp}(\sum_{j=1}^d \kappa_j)$.  At time $T + \Delta T$ a coordinate $i \in \beta$ is released with positive probability $\kappa_i / \sum_j \kappa_j$. Conditional on $\Delta T$ and on the event that the coordinate $i$ is released at time $T + \Delta T$, the sticky Zig-Zag behaves as a 1 dimensional ordinary Zig-Zag sampler until time $\tau_{\fS} + \Delta T + \min(\Delta T_1, \Delta T_2)$, where, similarly as before, $\Delta T_1 \sim \text{Exp}(\sum_{j\ne i} \kappa_j)$ (and it is independent from the trajectory of the free coordinate) and  $\Delta T_2$ is the hitting time to $0$ of the coordinate process $Z_i(\tau_{\fS} + \Delta T + t)_{t>0}$. By Proposition~\ref{prop: ergodicity of subzigzags}, $\Delta T_2$ is finite and $\PP(\Delta T_2 < \Delta T_1)>0$. By using the Markovian structure of the process and iterating this argument for all $i \in \beta$ we conclude that $\PP(\tau_{S} < \infty) = 1$.
\end{proof}

By \textcite[Theorem 6.1]{meyn1993stability}, the sticky Zig-Zag sampler is ergodic if it is Harris recurrent with invariant probability $\mu$ and if some skeleton of the chain is irreducible. For the latter condition, any skeleton $Z^{(\Delta)} = (Z(0)), Z(\Delta), Z(2\Delta), \dots)$ (with $\Delta > 0$) is irreducible relative to the measure $\overline \delta_0$ as the process, once it has reached the null model will  stay there for a time $\Delta T \sim\text{Exp}(\sum_{j=1}^d \kappa_j)$ and $\PP(\Delta T>\Delta) > 0$.

\subsection{Recurrence time of the sticky Zig-Zag to 0}\label{app: recurrence time to 0}
The recurrent time to the point $\bm 0 = (0,0,\dots,0)$ is derived with a simple heuristic argument. We assume the sticky Zig-Zag to have unit velocity components and to be ergodic with stationary measure $\mu$.  Clearly,  the expected time to leave $\bm 0$ is $(\kappa d)^{-1}$ since each coordinate leaves 0 according to an independent exponential random variable with parameter $\kappa$. Denote by $\tau_0$ the recurrent time to 0, i.e. the random time spent outside $\bm 0$ before returning to $\bm 0$. By ergodicity, the expectation of $\tau_0$ must satisfy the following equation
\[
\frac{(\kappa d)^{-1}}{\mu(\{\bm 0 \})} = \frac{\EE[\tau_0]}{1-\mu(\{\bm 0\})}.
\]

\clearpage 

\section{Other sticky PDMP samplers}\label{app: other sticky samplers} 
Here we extend the results presented in Section~\ref{sec: sticky PDMP samplers>Sticky Zig-Zag sampler} for two other Sticky PDMP samplers:  the sticky version of the Bouncy particle sampler (\cite{2015arXiv151002451B}) and the Boomerang sampler (\cite{bierkens2020boomerang}),  the latter having Hamiltonian deterministic dynamics invariant to a prescribed Gaussian measure.
 To visually assess the difference in sample paths, we show in Figure~\ref{fig:stickylook}  a typical realization of the Sticky Zig-Zag sampler, Sticky Bouncy particle sampler and Sticky Boomerang sampler.
 
\begin{figure}[ht!]
    \centering
\includegraphics[width=0.8\linewidth]{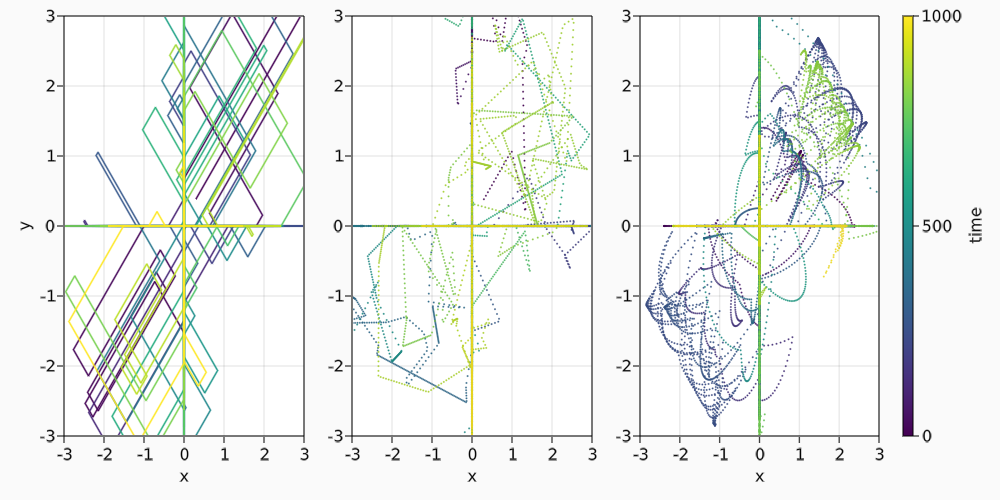}

    \caption{ $(x\text{-}y)$ phase portraits, of 3 different sticky PDMP samplers targeting the measure of Equation~\eqref{eq: target measure} with $\exp(-\Psi)$ being a mixture of two bivariate Gaussian densities centered respectively in the first and the third quadrant of the $x\text{-}y$ axes. Left: Sticky Zig-Zag sampler. Middle:  sticky Bouncy Particle sampler with refreshment rate equal to 0.1. Right:  sticky Boomerang sampler with refreshment rate equal to 0.1. For all the samplers, $\kappa_1=\kappa_2=0.1$ and the final clock was set to $T = 10^3$. As the sticky Bouncy Particle sampler and the Boomerang sampler don't have constant speed, we marked their continuous trajectories in the phase plots with dots. The distance of dots indicates the speed of traversal.} 
    \label{fig:stickylook}
\end{figure}

\subsection{Sticky Bouncy Particle sampler}\label{sec: Sticky Bouncy Particle sampler}
The inner product and the norm operator in the subspace determined by $A$ is denoted by $\langle x, v \rangle_A := \sum_{i \in A} x_i v_i$ and $\| x\|_A := \sum_{i \in A} x_i^2$ with the convention that $\langle \cdot,\cdot \rangle_{\{1,2,\dots,d\}} = \langle \cdot,\cdot \rangle$ and $\|\cdot\|_{\{1,2,\dots,d\}} = \|\cdot\|$. The deterministic dynamics of the sticky Bouncy Particle process are identical to that of  the Sticky Zig-Zag process, having piecewise constant velocity. For each $i \in \{1,2,\dots,d\}$, when the process hits a state $(x, v) \in \mathfrak F_i$, the $i$th coordinate $(x_i,v_i)$ sticks for an exponentially distributed time with  rate equal to $\kappa_i|v_i|$ while the other coordinates continue their flow until a reflection or refreshment event happens. A reflection occurs with an inhomogeneous rate equal to 
\[
 \lambda(x, v) = \max( 0, \langle v, \nabla \Psi(x) \rangle_{\alpha}), 
\]
where $\alpha$ is as defined in Equation~\eqref{eq:alpha}.
At reflection time the process jumps with a contour reflection of the active velocities   with respect to $\nabla \Psi$:
\[
(R_\Psi (x, v)v)_i = \begin{cases}
v_i & i \notin \alpha(x,v)\\
v_i - 2\frac{\langle \nabla \Psi(x), v\rangle_{\alpha}}{ \|\nabla \Psi(x)) \|^2_{\alpha}} \partial_i \Psi(x) & \text{else}.
\end{cases} 
\]
Similarly to the ordinary Bouncy Particle sampler, the sticky Bouncy Particle sampler refreshes its velocity component at exponentially distributed times with homogeneous rate equal to $\lambda_{\rm ref}$. This is necessary for avoiding pathological behaviour of the process (see \cite{2015arXiv151002451B}).
At  refreshment times, each coordinate renews its velocity component independently according to the following refreshment rule
\begin{equation}\label{eq: refresh rule bps}
v'_i \sim \begin{cases}Z_i & (x,v) \notin \mathfrak F_i, \\
                    \sign(v_i)|Z_i| & (x,v) \in \mathfrak F_i, \end{cases} 
\end{equation}
where $Z_i \stackrel{i.i.d.}{\sim} \mathcal{N}(0,1)$, independently of all random quantities. The refreshment rule coincides with the refreshment rule given in the ordinary Bouncy Particle sampler algorithm \textcite{2015arXiv151002451B} for the coordinates whose index is in the set $\alpha$. For the components which are stuck at $0$, the refreshment rule renews the velocity without changing its sign. This prevents the possibility for the $i$th stuck component to jump out the set $\mathfrak F_i$ (changing its label from frozen to active at refreshment time).   

The extended generator of the sticky Bouncy Particle sampler is given by
\begin{equation*}
    \cA f(x, v) =  \sum_{i=1}^d\mathcal{G}_i f(x, v) + \lambda(x,v)(f(x, R_\Psi (x,v) v) - f(x, v)) + \lambda_{\rm ref} \int \left(f(x, w) - f(x, v)\right)\varrho_{x,v}(w)\dd w
\end{equation*}
and
\begin{equation*}
    \mathcal{G}_i f(x, v) = \begin{cases}\!
 |v_i| \kappa_i \left(f(T_i(x, v)) - f(x, v)\right) & (x,v) \in \mathfrak F_i\\
 v_i \partial_{x_i} f(x, v)  & \text{else,}
\end{cases}
\end{equation*}
where
\[
\varrho_{x,v}(w) = \rho(w_{\alpha(x,v)})\prod_{i\in \alpha(x,v)^c}2\rho(w_i) \1
_{v_iw_i >0},
\]
for sufficient regular functions  $f\colon E\to \RR$ in the extended domain of the generator. Here, $\rho(y)$ is the standard normal density function evaluated at $y$.
\begin{proposition}\label{prop: sticky bps invariant measure}
  The $d$-dimensional sticky Bouncy Particle sampler is  invariant to the measure
    \begin{equation}
        \label{eq: target bps}
        \mu(\dd x, \dd v) =  \frac1C  \rho(v) \dd v \exp(-\Psi(x))\prod_{i = 1}^d\left(\dd x_i + \frac{1}{\kappa_i}\left( \1_{v_i>0} \delta_{0^-}(\dd x_i) + \1_{v_i<0} \delta_{0^+}(\dd x_i) \right)\right)
    \end{equation}
    for some normalization constant C.  
\end{proposition}
\begin{proof}
The transition kernel $R_\Psi (x)$ satisfies the following properties:
\begin{equation*}\label{eq: bps property 1}
    \langle \nabla \Psi(x), R_\Psi (x,v)v \rangle_{\alpha} = - \langle \nabla \Psi(x), v \rangle_{\alpha}
\end{equation*}
and 
\begin{equation*}\label{eq: bps property 2}
\| R_\Psi(x,v)v\|^2 = \| v\|^2_{\alpha^c} + \| R_\Psi(x,v)v\|^2_{\alpha} = \| v\|^2_{\alpha^c } + \|v\|^2_{\alpha} = \| v\|^2
\end{equation*}
so, $\rho(R_\Psi^A(x)v) = \rho(v)$ ($\rho(x)$ here denotes the standard Gaussian density evaluated at $x$). Furthermore $\lambda$ satisfies
\begin{equation}
    \label{eq: property 1 bps}
-\langle v, \nabla \Psi(x) \rangle_{\alpha} + \lambda(x,v) - \lambda(x, R_\Psi(x, v)v) = 0 , \quad \forall (x,v) \in E.
\end{equation}
Let us check that the process satisfies  $\int \Ge f(x,v) \mu(\dd x, \dd v) = 0$, for all $f \in D = \{f \in C^1_c( E), \cA f \in C_b( E)\}$ where $\cL$ is the extended generator $\cA$ restricted to $D$.

First let us fix some notation:  denote $f_i(y) = f(x[i\colon y], v)$, $Rf(x,v) = f(x, R_\Psi(x, v)v)$ and $R\lambda(x,v) = \lambda(x, R_{\Psi}(x, v)v)$. Also write  
$\delta_0(\dd x_i, v_i) :=\1_{v_i<0}\delta_{0^+}(\dd x_i) + \1_{v_i>0}\delta_{0^-}(\dd x_i)$ and $ \Delta_i f(x, v) := f(x[i\colon 0^+],v) -f(x[i\colon 0^-],v))$. We have this preliminary result:
\begin{align}
\int &  \sum_{i=1}^d\cG_i f \dd \mu = \frac1C \sum_i \int \left( \cG_i f 
 \exp(-\Psi(x)) (\dd x_i + \frac1{\kappa_i}\delta_0(\dd x_i)) \right)   \prod_{j \ne i}\left(\dd x_j + \frac{1}{
 \kappa_j}\delta_{0}(\dd x_j, v_j)\right)  \rho(v) \dd v \nonumber\\
&= \frac1C \sum_i \int \left(v_i \partial_{x_i} f \exp(-\Psi(x)) \dd x_i +  v_i \Delta_i f \exp(-\Psi(x)) \delta_0(\dd x_i) \right)
  \prod_{j \ne i}\left(\dd x_j + \frac{1}{
 \kappa_j}\delta_{0}(\dd x_j, v_j)\right)  \rho(v) \dd v \label{eq: step 2 sticky bps proof}\\
&= \frac1C \sum_i  \int \left(v_i \partial_{x_i} \Psi(x) f(x, v)  \exp(-\Psi(x)) \dd x_i \right)
  \prod_{j \ne i}\left(\dd x_j + \frac{1}{
 \kappa_j}\delta_{0}(\dd x_j, v_j)\right)  \rho(v) \dd v \label{eq: step 3 sticky bps proof}\\
&= \frac1C  \sum_{A\subset \{1,\dots,d\}} \left(\sum_{i \in A}\left( \int v_i \partial_{x_i} \Psi(x) f(x, v)  \exp(-\Psi(x)) \dd x_A \right) \prod_{j \in A^c, }\frac{1}{\kappa_j} \delta_0(\dd x_j, v_j) \right) \label{eq: step 4 sticky bps proof}\\
&= \frac1C  \sum_{A \subset \{1\dots, d\}}\int \langle v, \nabla  \Psi(x[A^c\colon 0]) \rangle_A f(x[A^c\colon 0], v) \exp(-\Psi(x[A^c\colon 0]))   \dd x_A  \prod_{j \in A^c}\frac{1}{\kappa_j} \rho(v) \dd v \nonumber
\end{align}
Here from \eqref{eq: step 2 sticky bps proof} to \eqref{eq: step 3 sticky bps proof} we used integration by parts in the two half planes $(\infty, 0^+]$ and $[0^-, -\infty)$. 
For the equivalence of \eqref{eq: step 3 sticky bps proof} to \eqref{eq: step 4 sticky bps proof} note that placing $|A|$ balls in $d$ numbered boxes and marking one of them (say the ball in box $i$) is equivalent to placing a marked ball in box $i$ and distributing the remaining unmarked balls over the remaining boxes. Also notice that
\begin{align*}
     \int \lambda_{\rm ref} &\int (f(x,w) - f(x, v)) \varrho(w)\dd w \dd \mu =   \\
        &\frac 1 C \sum_{A\subset \{1,2,\dots,d\}}\lambda_{ref} \int \left( f(x, w) - f(x, v)\right) \exp(-\Psi(x))\dd x_A  \\
        &\qquad \qquad \qquad \qquad \times\prod_{i \in A^c}\frac1{\kappa_i}\delta_{0^-}(\dd x_i) \1_{v_i> 0}\1_{w_i> 0} 2^{|A^c|} \rho(v)\rho(w) \dd v \dd w  \\
        &\quad + \frac1C \sum_{A\subset \{1,2,\dots,d\}}\lambda_{\rm ref} \int \left( f(x, w) - f(x, v)\right) \exp(-\Psi(x))\dd x_A  \\
        &\qquad \qquad \qquad \qquad \times\prod_{i \in A^c}\frac1{\kappa_i}\delta_{0^+}(\dd x_i) \1_{v_i <0}\1_{w_i <0} 2^{|A^c|} \rho(v)\rho(w)\dd v \dd w,  
\end{align*}
which is equal to 0 by symmetry between $v$ and $w$.
Then 
\begin{align}
    \int \mathcal{L}f \dd \mu &= \frac1C  \sum_{A \subset \{1\dots, d\}} \int \langle v, \nabla  \Psi(x[A^c \colon 0])\rangle_A  \exp(-\Psi(x[A^c\colon 0])) f(x[A^c \colon 0], v)  \dd x_A \prod_{j \in A^c}\frac{1}{\kappa_j} \rho(v) \dd v \nonumber\\
    &\quad + \int (\lambda (x, R_\Psi(x,v)) - \lambda (x, v)) f(x,v) \mu(\dd x, \dd v) \nonumber\\
    &=  \frac1C  \sum_{A \subset \{1\dots, d\}} \int \langle v, \nabla  \Psi(x[A^c:0]) \rangle_A \exp(-\Psi(x[A^c\colon 0])) f(x[A^c\colon 0], v)  \dd x_A  \prod_{j \in A^c}\frac{1}{\kappa_j}  \rho(v) \dd v \label{eq: bps invariant measure 1}\\
    &\quad + \frac 1C \sum_{A \subset \{1,\dots,d\}} \int \left( \lambda (x[A^c\colon 0],R_\Psi v) - \lambda (x[A^c\colon 0],v)\right)f(x[A^c\colon 0], v)\exp(-\Psi(x[A^c\colon 0])) \dd x_A \label{eq: bps invariant measure 2}\\ 
    &\quad \times \prod_{j \in A^c}\frac{1}{\kappa_j}  \rho(v) \dd v \nonumber\\ 
    &= 0, \nonumber
\end{align}
where in Equation~\eqref{eq: bps invariant measure 1}-\eqref{eq: bps invariant measure 2} we used a change of variable $v' = R_\Psi(x,v)v$ and property \eqref{eq: property 1 bps}.
\end{proof}
\begin{remark}\label{remark: general refreshment bps}
In more generality, the transition kernel at refreshment times can be chosen as follows: with two refreshment transition densities $q^A$ and $q^F$ such that $q^{A}(w_A\mid v_A)\rho(v_A)$ and $q^{F}(w_F\mid v_F)\rho(v_F)$ for each $A \sqcup F = \{1, \dots, d\}$
are symmetric densities in $w, v$, the refreshment kernel
\[
\varrho_{x,v}(dy,dw) = q^A(w_{\alpha(x,v)}\mid w_{\alpha(x,v)}) q^{F}(w_{\alpha^c(x,v)}\mid w_{\alpha^c(c,v)})\delta_{\cF (x,v,w)}(\dd y) \dd w
\]
where
\[(\cF (x,v,w))_i = \begin{cases}
0^- & \text{if }  x_i = 0^+, v_i < 0, w_i > 0,\\
0^+ & \text{if } x_i = 0^-, v_i > 0, w_i < 0,\\
x_i & \text{else}
\end{cases}\]
leaves the target measure $\mu$ invariant.
\end{remark}
The transition kernels given in Remark~\ref{remark: general refreshment bps} satisfy the Equation~$\lambda_{\rm ref} \int f(x, w) - x(x, v) \varrho_{x,w} \dd w \dd \mu = 0$ and therefore, by similar computations as in the proof of Proposition~\ref{prop: sticky bps invariant measure}, leave $\mu$ invariant. For example, the preconditioned Crank-Nicolson scheme \textcite{cotter2013mcmc} falls withing this setting.

\subsection{Sticky Boomerang sampler} 
The sticky Boomerang sampler has Hamiltonian dynamics prescribed by the vector field $\bar \xi_i(x_i, v_i) = (v_i, -x_i)$ with close-form solution 
\begin{equation} \label{eq: boom det dynamics}
(x_i(t), v_i(t)) = (\cos(t)x_i(0) + \sin(t)v_i(0), -x_i(0)\sin(t) + \cos(t)v_i(0)),
\end{equation}
and is invariant to a prescribed Gaussian measure centered in 0. Define $U(x)$ such that
\[
U(x) = \Psi(x) - \frac12 x'\Sigma^{-1}x 
\]
for a positive semi-definite matrix $\Sigma \in \RR^{d\times d}$. Consider for example the application in Bayesian inference with spike-and-slab prior (Equation~\eqref{eq: spike-and-slab prior}) where $\{\pi_{i}\}_{i=1}^d$ are centered Gaussian densities with variance $\sigma_i^2$. Then a natural choice is $\Sigma = \text{Diag}(\sigma_1^2,\sigma^2_2,\dots,\sigma^2_n)$.

Similarly to the sticky Bouncy Particle sampler, the process reflects its velocity at an inhomogeneous rate given by 
\[
\lambda(x, v) = \langle v, \nabla U(x)\rangle_{\alpha}^+
\]
with reflection specified by the  transition kernel
\[
(R_U (x,v)v)_i = \begin{cases}
v_i & i \notin \alpha\\
v_i - 2\frac{\langle \nabla U(x), v\rangle_{\alpha}}{ \|\nabla \Sigma^{1/2}U(x)), \|^2_{\alpha}} \langle \Sigma_{[i, :]}, \nabla U(x) \rangle_{\alpha} & \text{else}
\end{cases} 
\]
and refreshes the velocity at exponentially distributed times with rate equal to $\lambda_{\rm ref}$ according to the rule given in Equation~\eqref{eq: refresh rule bps}.
\begin{proposition}\label{prop: sticky boom invariant measure}
  The $d$-dimensional sticky Boomerang sampler is  invariant to the measure in Equation~\eqref{eq: target bps}.
\end{proposition}
\begin{proof}
The extended generator of the sticky $d$-dimensional Boomerang process is given by
\begin{equation*}
    \cA f(x, v) =  \sum_{i=1}^d\mathcal{G}_i f(x, v) + \lambda(x,v)(f(x, R_U (x,v) v) - f(x, v)) + \lambda_{\rm ref} \int \left(f(x, w) - f(x, v)\right)\varrho_{x,v}(w)\dd w
\end{equation*}
and
\begin{equation*}
    \mathcal{G}_i f(x, v) = \begin{cases}\!
 |v_i| \kappa_i \left(f(T_i(x, v)) - f(x, v)\right) & (x,v) \in \mathfrak F_i\\
 v_i \partial_{x_i} f(x, v)  +  x_i \partial_{v_i} f(x, v)& \text{else,}
\end{cases}
\end{equation*}
where
\[
\varrho_{x,v}(w) = \rho(w_{\alpha(x,v)})\prod_{i\in \alpha(x,v)^c}2\rho(w_i) \1
_{v_iw_i >0},
\]
$\rho(y)$ being the standard normal density function evaluated at $y$ and 
for sufficient regular functions  $f\colon E\to \RR$ in the extended domain of the generator. Then, define $D = \{f \in C^1_c( E), \cA f \in C_b( E)\}$ and $\cL$ as the extended generator $\cA$ restricted to $D$. 
The component of the extended generator $(x,v) \to \partial_{x_i} f(x, v)  +  x_i\partial_{v_i} f(x, v)$  produces Hamiltonian dynamics (see Equation~\eqref{eq: boom det dynamics})
preserving any Gaussian measure centered on $0$.
Notice that the $R_U (x)$ satisfies 
\[
\langle \nabla U(x), R_U(x)v \rangle_{\alpha(x,v)} = - \langle \nabla U(x), v \rangle_{\alpha(x,v)}
\]
and that 
\[
\| \Sigma^{-1/2}R_U(x)v\| = \|\Sigma^{-1/2} v\|. 
\]
Then one can check that $\int \Ge f(x,v) \mu(\dd x, \dd v) = 0$ by carrying out similar computations as in the proof of Proposition~\ref{prop: sticky bps invariant measure}.
\end{proof}
A variant of the sticky Boomerang sampler is the sticky factorised Boomerang sampler (being the sticky version of the factorised Boomerang sampler introduced in \cite{bierkens2020boomerang}). Here the process has the same dynamics, refreshment rule and sticky events of the sticky Boomerang process but has a different reflection rate and reflection rule. Similarly to the Sticky Zig-Zag process, the first reflection time of the sticky factorised Boomerang sampler is given by the minimum of $|\alpha(x,v)|$ Poisson times  $\{\tau_j \colon j \in \alpha(x,v)\}$ with $\tau_j \sim \text{Poiss}(t \rightarrow \lambda_j(\varphi(t,x,v))$ and $\lambda_j(x,v) = (\partial_{x_j}U(x)v_j)^+$. Likewise the Sticky Zig-Zag process, at the reflection time the process reflects its velocity by changing the sign of the $i$th component $v 
\to v[i\colon -v_i]$ where  $i = \argmin\{\tau_j \colon j \in \alpha(x,v) \}$. As shown in \textcite{bierkens2020boomerang} the factorised Boomerang sampler can outperform the Boomerang sampler when $\partial_{x_i} U$ is function of few coordinates.

\clearpage

\section{Comparison between reversible jump PDMPs and sticky PDMPs}\label{app: Comparison between reversible jump PDMPs and sticky PDMPs}
In this appendix, we discuss the differences between the sticky PDMPs and  RJ (Reversible Jumps) PDMPs presented in \textcite{chevallier2020reversible} which, similarly to us, addresses variable selection problems using PDMP samplers.

The approach taken in \textcite{chevallier2020reversible} is  based on  the framework of reversible jump (RJ) MCMC as proposed in  \textcite{Green95reversiblejump} and its derivation is therefore substantially different from our approach. Nonetheless, the  samplers have certain similarities.  The dynamics of both the RJ PDMPs in \textcite{chevallier2020reversible} and the sticky PDMPs proposed in this paper allow each coordinate to  stick at 0 for an exponential time. The rate of the exponential time of the sticky PDMPs  depends only on the velocity component of each coordinate, while the rate of RJ PDMPs can depend on the current state of the process. The latter is slightly more general as it allows to choose freely a prior weight on the Dirac measure for each possible model (while our approach allows to choose freely a prior weight on the Dirac measure of each possible coordinate). An important difference between the two methods is the behaviour of the process after the particle sticks at 0:  the velocity of the coordinate of the sticky PDMPs is restored to its previous value while for RJ PDMPs, a new velocity is drawn independently to the previous one. The former action introduces non-reversible jumps between models while the latter reversible jumps and a random walk behaviour when jumping between models. This simple, yet substantial, difference leads to two different limiting behaviour of the two processes when the number of Dirac measures increases. The limiting behaviour of both processes is unvelied  below in Appendix~\ref{app: limiting behaviour} through numerical simulations:  while the Sticky Zig-Zag converges to ordinary Zig-Zag, the RJ Zig-Zag asymptotically exhibits diffusive behaviour.

For RJ PDMPs, the random walk behaviour is mitigated  by introducing a tuning parameter $p$ which allows each coordinate to stick at 0 only a fraction of times when hitting 0 (and compensating for this by down-scaling the rate of the exponential waiting time when the coordinate sticks). The parameter $p$ is tuned to be equal to  $0.6$ based on empirical criteria. In Appendix~\ref{app: heuristic} we investigate the possibility to introduce the tuning parameter $p$ in the Sticky Zig-Zag sampler and, based on a heuristic argument and a simulation study, we concluded that it is not beneficial for us.

\subsection{Heuristics for the choice of \texorpdfstring{$p$}{p}}\label{app: heuristic}
Here we investigate the possibility of introducing the parameter $p$ to the Sticky Zig-Zag sampler. This parameter was originally introduced in \textcite{chevallier2020reversible}.  Based on the heuristic argument and the simulation study given below, we conclude that the introduction of $p$ does not improve the performance of the Sticky Zig-Zag sampler.

The parameter $p$ defines the probability for a coordinate to stick at 0 when it hits 0. By introducing this parameter, the times of the particles stuck at 0 has to be rescaled by a factor of $p$ in order target the right measure. 

Consider a trajectory $\{z_t \colon 0<t<T\}$ of the one dimensional ordinary Zig-Zag sampler (without stickiness) targeting a given measure. In this case, one could create a trajectory of the Sticky Zig-Zag process retrospectively just by adding constant segments equal to 0, every time the process hits 0 with random length equal to $XY$, with $X \sim \mathrm{Ber}(p)$ and $Y \sim \mathrm{Exp}(\kappa/p)$, $X$ independent from $Y$. Then, if the trajectory $z_t$ hits 0 $N$-times, the total occupation time of the sticky process in 0 is Gamma-distributed with shape parameters $\frac{N}{p}$ and inverse scale parameter $p\kappa$ (in variable selection, this would correspond to the posterior probability of the sub-model without the coefficient). While the mean of this random variable is constant for every $p$, its variance is $\frac{N}{\kappa p}$ and is minimized when $p = 1$.

Based on the aforementioned heuristics, it appears not useful to introduce the parameter $p$ for the Sticky Zig-Zag. This claim is supported by simulations presented in Figure~\ref{fig:varying_p}, where we vary $p$ from 0.1 (top) to 1.0 (bottom) for a 20 dimensional Gaussian density with pairwise correlation equal to 0.99 and relative to the measure 
\begin{equation}
\label{eq: test reference measure}
\prod_{i=1}^{d} \big(  \mathrm{d}x_i +  c\sum_{j \in \mathbb{N}} \delta_{j*0.01}(\mathrm{d}x_i) \big),
\end{equation}
with $c = 1.0$. In Figure~\ref{fig:varying_p}, left panels, the traces are more erratic when $p$ is small and the process traverses the space in less time when $p$ is large
(notice the different ranges of the vertical axis). In Figure~\ref{fig:varying_p}, right panels, the phase portrait of the first two coordinates is shown. By visual inspection it is possible to notice that the phase portrait fails to be symmetric on the axis $x_1 = -x_2$ for $p$ small while it succeeds for $p = 1$ (notice again the different ranges of the axes), hence suggesting that Zig-Zag sampler has a better mixing for $p = 1$.
\begin{figure}[ht!]
    \centering
    \includegraphics[width=0.8\linewidth]{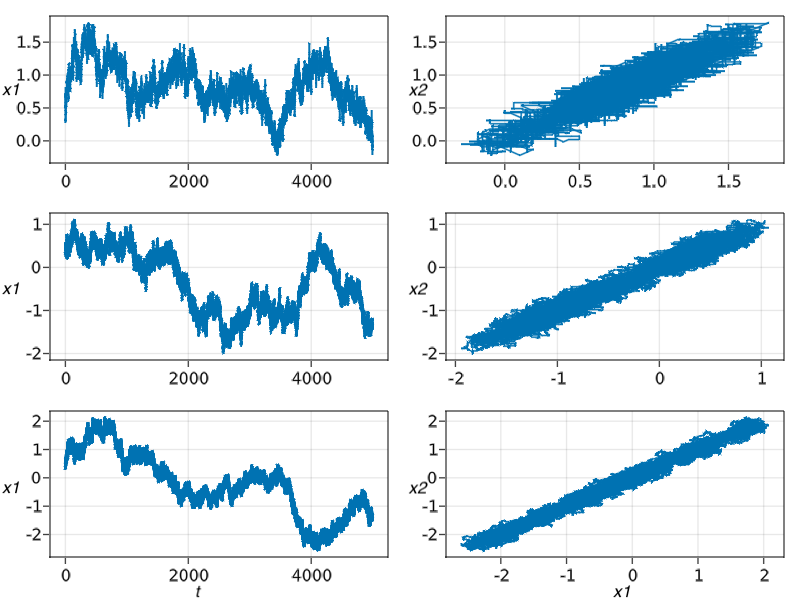}
    \caption{$x_1$ trace plots (left) and $x_1$-$x_2$ phase portraits (right) of the Sticky Zig-Zag samplers with final clock $T = 50^3$ with $p$ equal to  $0.1$ (top), $0.5$ (center), $1.0$ (bottom). The target measure has a Gaussian density with pairwise correlation equal to 0.99 relative to the reference measure of Equation~\eqref{eq: test reference measure}. By comparing the symmetry of the empirical measures along the diagonal and the range of the coordinates, one can conclude that the algorithm performs best for $p = 1$.}
    \label{fig:varying_p}
\end{figure}

\subsection{Limiting behaviour}\label{app: limiting behaviour}
 Here we show the different limiting behaviour between the RJ-PDMP samplers and the sticky PDMP samplers as the number of Dirac measures increases. 
 
 The limiting behaviour of the two samplers significantly differ because after every time a coordinate  sticks at a point mass, the sticky PDMP sampler preserves the velocity component while RJ PDMP sampler has to refresh a new independent velocity. We illustrate the limiting behaviour of the two samplers through simulations where we let the Sticky Zig-Zag and the RJ-Zig-Zag sampler (with $p = 0.6$) target a 20-dimensional measure with a Gaussian density with pairwise correlation equal to 0 (Figure~\ref{fig:revjump_vs_sticky1}) and 0.99 (Figure~\ref{fig:revjump_vs_sticky2}) relative to the reference measure of Equation~\eqref{eq: test reference measure} with $c = 10$. While the Sticky Zig-Zag sampler resemble an ordinary Zig-Zag sampler, the RJ-PDMP sampler has a limiting diffusive behaviour and appears to explore the space less efficiently than the sticky PDMP sampler (see the range of the axes and the symmetries of the measure around the axis $x_2 = -x_1$ ). 

\begin{figure}[h]
    \centering
\includegraphics[width=0.75\linewidth]{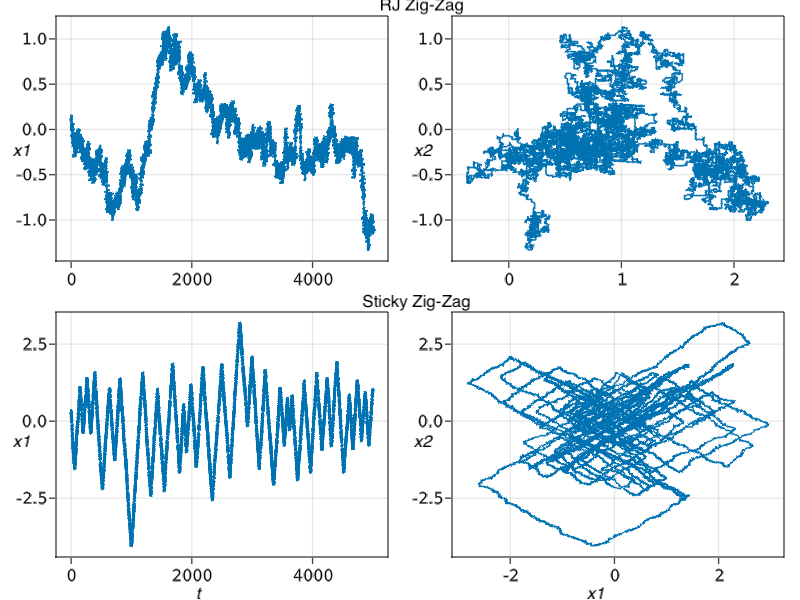}
    \caption{Comparison between RJ Zig-Zag samplers (first row) and Sticky Zig-Zag samplers (second row) targeting a 20 dimensional measure with Gaussian density with pairwise correlation equal to 0.0  and relative to the reference measure in Equation~\eqref{eq: test reference measure}. Column 1: trace plot of the first coordinate. Column 2: trace plot of the second column. In all cases $T = 10^4$. By looking at the range of each coordinate, it is clear that the Sticky Zig-Zag mixes faster than its reversible counterpart.}
    \label{fig:revjump_vs_sticky1}
    \end{figure}
    
    \begin{figure}[h]
    \centering
   \includegraphics[width=0.75\linewidth]{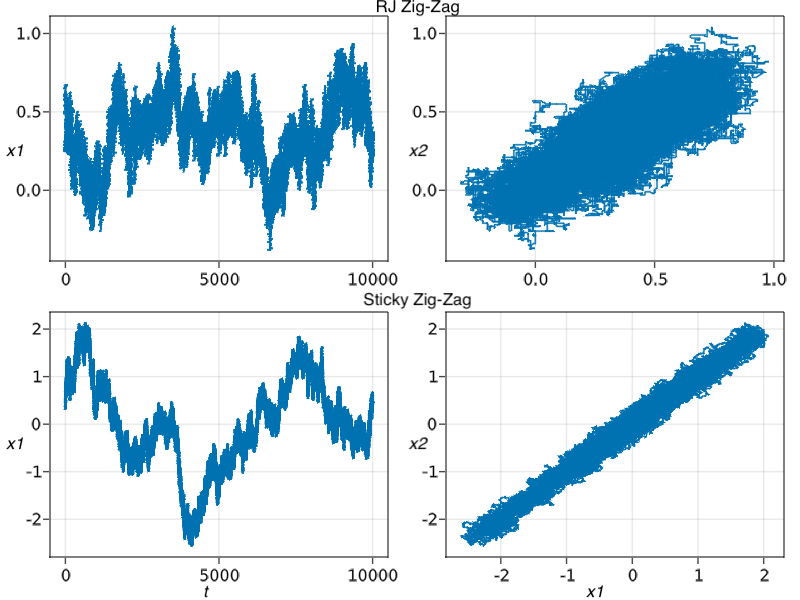}
        \caption{Same description as in Figure~\ref{fig:revjump_vs_sticky1}, except now for a Gaussian measure with pairwise correlation equal to 0.99. By looking for example at the symmetry along the axis $x_2 = -x_1$ and the ranges of the coordinates, it is clear that the Sticky Zig-Zag outperforms the RJ Zig-Zag.}
    \label{fig:revjump_vs_sticky2}
\end{figure}

\clearpage

\section{Details of Section~\ref{sec:runtime analysis}}\label{app: details of section 4}
\subsection{Bayes factors for Gaussian models}\label{app: bayes_factors}
Let $ (X,Y) \sim N(\mu, \Gamma^{-1})$, written in block form as
\[
\mu = \begin{bmatrix}\mu_x\\ \mu_y\end{bmatrix}, \quad \Gamma = \begin{bmatrix}\Gamma_{x} & \Gamma_{xy} \\ \Gamma_{xy}' &\Gamma_{y} \end{bmatrix}.
\]
Denote the density of $(X,Y)$ evaluated at $(x,y)$ by $\phi([x,y]; \,\mu, \,\Gamma^{-1})$.
Let
\begin{equation}
X\mid (Y=y) \sim \cN(\mu_{x|y},\, \Gamma_{x}^{-1}) \label{lemma:conditional density}
\end{equation}
be the marginal density of $X$ given $Y = y$, where $\mu_{x\mid y} = \mu_x - \Gamma^{-1}_x \Gamma_{xy}(y - \mu_y)$. 
Assume $\Gamma_x$ to be positive definite and let the marginal density of $Y$ be
\begin{equation}
\int\phi([x,y]; \,\mu, \,\Gamma^{-1}) \dd x 
= (2\pi)^{\tfrac{d_x - d}2}|\Gamma|^{\frac12} |\Gamma_x|^{-\frac12}   \exp\left(\frac12 \mu_{x\mid y}'\Gamma_x \mu_{x\mid y}
-\frac12[ - \mu_x, y - \mu_y]' \Gamma [ - \mu_x, y - \mu_y]
\right) \label{lemma:marginal density}
\end{equation}
where $d_x$ is the size of $X$.

We are now ready to compute the corresponding Bayes factors of two neighbouring (sub-)models as in Equation~\eqref{eq; Bayes_factors} when $\Psi$ is a quadratic function. For every set of indices $\alpha \subset \{1,2,\dots,d\} $ and for every $j$, the Bayes factors relative to two neighbouring (sub-)models  (those differing by only one coefficient) for a measure as in Equation~\eqref{eq: target measure} are given by
 \begin{equation}\label{eq: bayes_fact_appendix}
 B_j(\alpha) =  \frac{\mu( \cM_{\alpha \cup \{j\}})}{\mu(
\cM_{\alpha \setminus\{ j\}})} = \frac{\kappa_i\int_{\RR^{|\alpha \cup \{j\}|}} \exp(-\Psi(y)) \dd x_{\alpha \cup \{j\} }}{\int_{\RR^{|\alpha \setminus \{j\}|}} \exp(-\Psi(z)) \dd x_{\alpha \setminus \{j\}}}, 
 \end{equation}
   where $y = \{x \in \RR^d \colon x_i = 0,\, i \notin (\alpha \cup \{j\}) \}$, $z = \{x \in \RR^d \colon x_i = 0, \, i \notin (\alpha \setminus \{j\})]$. Since $\Psi$ is quadratic, we can write $\exp(-\Psi(x)) = C\phi(x; \,\mu, \,\Gamma^{-1})$ for some parameters $C,\mu, \Gamma$. By using both Equation~\ref{lemma:conditional density} and Equation~\ref{lemma:marginal density} we have that the right hand side of Equation~\eqref{eq: bayes_fact_appendix} is equal to 
 \[
 \kappa_i\sqrt{\frac{2\pi|\Gamma_{x_1}|}{|\Gamma_{x_2}|}}\exp\left(\frac{1}{2} (\mu'_{x_1\mid y_1 = \bm{0}} \Gamma_{x_1} \mu_{x_1\mid y_1 = \bm{0}} - \mu'_{x_2\mid y_2 = \bm{0}} \Gamma_{x_2} \mu_{x_2\mid y_2 = \bm{0}}) \right)
 \]
where $x_{1} = x_{\alpha_{-j} \cup \{j\}},\,x_{2} = x_{\alpha_{-j} \setminus \{j\}}$, $y_{1} = x_{\alpha_{-j}^c \setminus \{j\}},\, y_{2} = x_{\alpha_{-j}^c \cup \{j\}}$.
Furthermore, by Equation~\ref{lemma:conditional density}, the random variable at step 2 of the Gibbs sampler presented in Section~\ref{sec: gibbs sampler} can be simulated as $X_\alpha | (X_{\alpha^c} = \bm{0}) \sim \cN(\mu_{x_{\alpha} \mid x_{\alpha^c} = \bm{0}}, \Gamma_{x_{\alpha}})$.
\subsection{Simulating sticky PDMPs and sticky Zig-Zag samplers}\label{app: simulating sticky PDMPs}
Sticky samplers can be implemented recursively by modifying appropriately the ordinary PDMP samplers so to include sticky events as introduced in Section~\ref{sec: sticky PDMP samplers}. We discuss how to integrate local implementations of the algorithms to increase the sampler's performance  in case of a sparse dependence structure in the target measure and in case of local upper bounding rates.

Although PDMPs have continuous trajectories, the algorithm computes and saves only a finite collection of points (which we refer to as the skeleton of the continuous trajectory) corresponding to the positions, velocities and times  where the deterministic dynamics of the process change. In between those points, the continuous trajectory can be deterministically interpolated.

In case the $i$th partial derivative of the negative score function is a sum of $N_i$ terms, which is  the case for example in regression problems, subsampling techniques can be employed as described in Section~\ref{sec: Subsampling}.

\subsubsection{Computing Poisson times for PDMPs}\label{app: computing Poisson times}
As PDMPs move deterministically (and with simple dynamics) in between event times, the main computational challenge consists of simulating those times.
Given an initial position $(x, v)$, the distribution of the time until the next event is specified in \eqref{eq:jumptimes}. A sample from this distribution can be found by solving for $\tau'$ in the equation
\begin{equation}
    \label{eq: find the root lambda}
    \int_0^{\tau'} \lambda(\varphi(s, x, v)) \dd s = t, \quad t = \text{Exp}(1). \end{equation}
We then write that $\tau' \sim \text{Poiss}( \lambda(\varphi(\cdot,x,v))$. When it is not possible to find the root of Equation~\eqref{eq: find the root lambda}  in closed form, it suffices to find upper bounds $\overline \lambda$ for the rate functions which satisfies, for any $(x,v)\in E$ and for 
some $\Delta = \Delta(x,v) > 0$
\begin{equation}\label{eq: lambda bound}
\overline \lambda(t, x, v)\ge \lambda(\varphi(t, x,v)), \quad \Delta \ge t  \ge 0,
\end{equation}
for which this is possible and use the thinning scheme: Let $\tau' \sim \text{Poiss}(\bar\lambda(\cdot,x,v))$; if $\tau' > \Delta$ then the proposed time is rejected and a new time has to be drawn as $\tau' \sim  \text{Poiss}(\bar \lambda(\cdot, \phi(\Delta, x,v)))$. We \textit{accept} the proposed time with probability $\lambda(\phi(\tau', x,v))/\bar\lambda(\tau', x,v)$. This scheme is referred as \emph{adaptive thinning} in \textcite{2015arXiv151002451B}.
More sophisticated and potentially efficient thinning  schemes have been proposed, see \textcite{sutton2021concave}. The simulation of unfreezing times is easier: 
once the $i$-th component hits zero then it sticks at zero for a time that is exponentially distributed with parameter $\kappa_i |v_i|$.

For the ordinary $d$-dimensional Zig-Zag and the factorised Boomerang sampler (these samplers are called  \emph{factorised PDMPs} in \cite{bierkens2020boomerang}), the reflection time is factorised as the minimum of  $d$ independent clocks $\tau_1,\tau_2,\dots,\tau_d$ where $\tau_i \sim \text{Poiss}( \lambda_i(\varphi(\cdot,x,v))$
for $i=1,2,\dots,d.$ The first reflection time of the $d$-dimensional sticky factorised samplers is obtained instead by finding the minimum of $|\alpha|<d$ independent clocks with the same rates $\lambda_i$ of the ordinary factorised sampler, but only for the active coordinates $i \in \alpha(x,v)$.

If $\partial_{x_i} \Psi$ (an estimate of $\partial_{x_i} \Psi$ when using subsampling) or the upper bound $\overline \lambda$ depends on fewer coordinates, then the evaluation of each reflection time is cheaper. 
The fully local implementation presented in \textcite{bierkens2020piecewise} exploits these two features once in proposing the reflection time and once for deciding whether to accept. Below, we discuss in more details the algorithm of Sticky Zig-Zag sampler with local upper bounds and with subsampling. 
\subsubsection{Local implementation:}\label{app: local implementation}
Assume that the sets $\overline A_i$ and  $\overline \lambda_i$ are such that 
\[
\overline \lambda_i(t, x, v) = f_i(t, x_{\overline A_i}), \quad \forall x, \text{ for } i = 1,2,\dots, d
\]
for some $f_i \colon \RR^+\times \RR^{|\overline A_i|} \to \RR^+$ with $\overline A_i \subset \{1,2,\dots, d\}$. 
Given an initial position $(x,v)$ and random times $\tau_j \sim \text{Poiss}(t \to \overline \lambda_j(t,x,v))$,  for $i \in \alpha$, denote by $i= \argmin_{j \in \alpha(x,v)}\,\tau_j$ and $ \tau = \min_{j \in \alpha(x,v)}\,\tau_j$ the first proposed reflection time. According to the thinning procedure for Poisson processes, the process flips the $i$th coordinate with probability $\lambda_i(\varphi(\tau, x, v))/ \overline \lambda_i(\tau, x,v)$. If the process flips the $i$th velocity, then the Poisson rates $\{\overline \lambda_j \colon j \in \alpha,\, \overline A_j  \niton i\}$ continue to be valid upper bounds so that the corresponding reflection times do not need to be renewed (see \cite[Section 4]{bierkens2020piecewise}, for implementation details).

In general, when the $i$th particle freezes at 0 or was stuck at 0 and gets released, the reflection times $\{\tau_j\colon i \in \overline A_j\}$ have to be renewed. 
However this is not always the case, as there are applications, such as the one in Section~\ref{subsec: examples>logistic regression}, for which the upper bounding rates $\{\overline \lambda_i\}_{i=1}^d$ continue to be valid upper bounds when one or more  particles hit 0 and therefore the waiting times computed before the particles hit 0 are still valid. 

\subsubsection{Fully local implementation:}\label{app: fully local implementation}
Consider now the decomposition of $\partial_{x_i} \Psi, \, i = 1,2,\dots, d$ given in Equation~\eqref{eq: decomposition partial derivatives} and such that  
\[
S(x,i,j) = f_{i,j}(x_{\tilde A_{i,j}}), \quad \forall x, \text{ for } (i,j) \in \{1,2,\dots, d\}\times\{1,2,\dots, N_i\}
\]
for some $f_{i,j} \colon \RR^{|\tilde A_{i,j}|} \to \RR$ with $\tilde A_{i,j} \subset \{1,2,\dots, d\}$.

The fully local implementation of the Sticky Zig-Zag with subsampling profits from local upper bounds and local gradient estimators by assigning an independent time for each coordinate, thus evolving the flow of only the coordinates which are required at each step and by stacking $\{\tau_j \wedge \tau^\star_j \colon j \in \alpha\}$, with $\tau_j$ being a proposed reflection time and $\tau_j^\star$ the hitting time to 0, and the unfreezing times $\{\tau_j^\circ \colon j \in \alpha^c\}$ in an ordered queue. For a documented implementation, see \textcite{mschauer/ZigZagBoomerang.jl}. 

Given an initial point $(x,v)$ and if $i = \argmin(\tau_j \colon j \in \alpha(x,v))$ is the coordinate of the first proposed reflection time $\tau = \min(\tau_j \colon j \in \alpha(x,v))$, the sampler reflects the velocity of the $i$th coordinate with probability ${\tilde \lambda_{i, J}(x_{\tilde A_i}(\tau), v)/\overline \lambda(\tau, x, v)}$ with $J \sim \text{Unif}(\{1,2,\dots, N_i\})$.  Hence, it is only  required to update the position of the coordinates with index in $ \tilde A_{i,J} \setminus \alpha^c(x,v)$. Then,
\begin{itemize}
    \item if the $i$th velocity flips, then the algorithm needs to update only the waiting times $\{\tau_j \colon j \in \alpha, \overline A_j \ni i \}$ (as described in Appendix~\ref{app: local implementation}) and, to this end, needs to update the position  of the coordinates with index $\{k \in \overline A_j\setminus \alpha^c(x,v) \colon i \in \overline A_j\}$; 
    \item in the other case, when the $i$th velocity does not change (shadow event), only $\tau_i$ has to be renewed so that only the particles in  $\overline A_i$ have to be updated.
\end{itemize}

\begin{remark}\label{rmk: sparse implementation} (Sparse implementation.)
 When the dimensionality $d$ is large, inserting each waiting time in a ordered queue and initializing the state space
 can be computationally expensive. If for example the product $k_i|v_i|$ is equal for all $i$, an alternative efficient and sparse implementation is possible. Here we simulate the sticky time for each frozen coordinate by means of simulating  the overall sticky time from the exponential distribution with rate $\sum_{i \in \alpha^c} \kappa_i|v_i|$ (which has to be renewed every time a new particle sticks at $0$) and selecting the particle to unfreeze uniformly from the set $\alpha^c$. A further improvement can be obtained by representing $x$ as a sparse vector and saving only the location of the active particles $\{x_i \colon i \in \alpha\}$.
\end{remark}

\subsection{Runtimes of the algorithms}\label{app: runtimes of the algorithms}

We will now compute typical runtimes for the Gaussian model, assuming a decomposition
\[ \Psi(x) = (x-\mu)' \Gamma(x-\mu) = \sum_{i=1}^N (x-\mu_i)' \Gamma_i (x-\mu_i) + c,\]
so that $N$ captures the dependence on the number of observations in a Bayesian setting.

\subsubsection{Sticky Zig-Zag sampler:} The computational cost of simulating PDMP samplers is intimately related with the number of random times generated. This, in turn, depends on the intensity of the rate $\lambda$ of the underlying Poisson process. For any initial position and velocity $(x,v)$,  the total rate of the Sticky Zig-Zag sampler is equal to
\begin{equation}
    \label{eq: overall rate}
    \lambda(x,v) = \sum_{i \in \alpha} \lambda_i(x,v) + \sum_{i \in \alpha^c} |v_i|\kappa_i  
\end{equation}
where, as before, $\alpha = \{i \colon x_i \ne 0\}$. In the following analysis, we drop the dependence on $(x,v)$  and we assume that the size of $\alpha(t) := \{i \colon x_i(t) \ne 0\}$ fluctuates around a typical value $p$ in stationarity. Thus $p$ represents the number of non-zero components in a typical model, and can be much smaller than $d$ in sparse models.

We consider the sticky  Zig-Zag with local implementation as in Remark~\ref{rmk: sparse implementation}
where we assume $\kappa := \kappa_1 = \kappa_2 = \dots = \kappa_{d}$.
We ignore logarithmic factors, e.g., for priority queue insertion. In the analysis below we distinguish between the computational costs of reflection events and unfreezing events.

The number of reflection and unfreezing events per unit time interval are respectively $\cO(p)$ and $\cO((d-p)\kappa)$
per unit time; see Equation~\eqref{eq: overall rate}. Once either a reflection or unfreezing event happens, we have to recompute between $\cO(1)$ and $\cO(p)$ new reflection event times (depending on the elements of $\overline A_i \cap \alpha$; see Appendix~\ref{app: local implementation}).
Finally, each newly computed reflection event time for the particles $i \in \alpha$ requires a computation ranging
 from $\cO(1)$ to $\cO(N)$.
The complexity $\cO(1)$ can be achieved using the subsampling technique (Section~\ref{sec: Subsampling})  in ideal scenarios (\cite{bierkens2019}). 
Table~\ref{tab: 1} in Section~\ref{sec:runtime analysis} summarizes the overall scaling complexity of the Sticky Zig-Zag algorithm for the quantities $p$ and $N$. 
\subsubsection{Gibbs sampler:}
At each iteration, the Gibbs sampler algorithm requires the evaluation of the Bayes factors which involves the inversion of a square matrix of dimension $p\times p$. This can be efficiently obtained with a Cholesky decomposition of a sub-matrix of $\Gamma$. This is a computation of  $\cO(p^3)$ when $\Gamma$ is full; a lower order is possible when $\Gamma$ is sparse. For example, in the example in Section~\ref{sec: example>Spatially structured sparsity}, the complexity of this operation is $\cO(p^{3/2})$. 
This is followed by  computing  sufficient statistics in step 2 of Section~\ref{sec: gibbs sampler} which involves the inversion of a triangular matrix which is  $\cO(|\alpha^2|)$ ($\cO(1)$ if the Cholesky factor is sparse) in addition to  an operation of order $pN$ (for example in linear or logistic regression). It is important to notice that if $\Gamma$ is sparse, its Cholesky factors might not be.
Our finding are summarized in Table~\ref{tab: 1} in Section~\ref{sec:runtime analysis} and validated by the numerical experiments of Section~\ref{sec: examples} (Figure~\ref{fig: runtime of heart}, Figure~\ref{fig:logistic_comparison}).

 \subsection{Mixing}\label{app: mixing}
 Next to the complexity per iteration, we should also understand the time the underlying process needs to explore the state space and to reach its stationary measure.
 Given the different nature of dependencies of the two algorithms, a rigorous and theoretical comparison of their mixing times is difficult. We therefore provide a heuristic argument for two specific scenarios.
 
 Let both algorithms be initialized at $x \sim \cN_d(0, I)$  with all non-zero coordinates ($\alpha^c = \emptyset$) and assume that the target $\mu$ assigns most of its probability mass to the null model $\cM_{\emptyset}$. Consider the following scenarios:
 \begin{itemize}
     \item \emph{A measure supported in every model}  and such that for any two models $\cM_{\alpha_i}$ and $\cM_{\alpha_j}$ with $\alpha_i\ne \alpha_j$, we have $\mu(\cM_{\alpha_i}) > \mu(\cM_{\alpha_j})$ if $|\alpha_i| < |\alpha_j|$. The Sticky Zig-Zag will be directed to the null model, each coordinate with speed 1, so that the first visit of the null set happens with an expected time  $\cO(\max_{i}(|x_i|))$ which is of $\mathcal O(\log d)$ if $x$ is standard Gaussian. On the other hand, the Gibbs sampler, at every iteration, randomly picks a coordinate and, if this is a non-zero coordinate, succeeds to set that coordinate to zero.  Denote by $\tau_\alpha$ the (random) number of iterations needed for the algorithm to set any non-zero coordinate to zero, when exploring a model $\cM_{\alpha}$. Then $\EE(\tau_\alpha)= d/|\alpha|$ which ranges from 1  (when $\cM_{\alpha}$ is the full model) to $d$ (for any sub-model with only one non-zero coordinate). Consider any sequence $\cM_{\alpha_{1}}, \cM_{\alpha_{2}},\dots,\cM_{\alpha_{d-1}}$ of models with $|\alpha_j| + 1 = |\alpha_{j+1}|$ (decreasing size) and with $\cM_{\alpha_1}$ begin the full model. By adding the expected number of iterations at each of those model, we conclude that the process started at $x$ in the full model, is expected to reach the null model in $\sum_{i=1}^d d/i$ iterations which is of $\cO(d\log(d))$.
     \item \emph{A measure supported on a single nested sequence of sub-models,} up to the full model: i.e. for a model $\cM_{\alpha_j}$, with $\mu(\cM_{\alpha_j}) \ne 0$ there is only one sub-model $\cM_{\alpha_i} \subset \cM_{\alpha_j}$ with $|\alpha_i|+1 = |\alpha_j|$  and  the smaller model again has more probability mass $\mu(\cM_{\alpha_i}) > \mu(\cM_{\alpha_j})$. By a similar argument as above, the first expected visit time of the null model is of $\cO(\sum_{i=1}^d |x_i|) = \mathcal O(d)$ for the Sticky Zig-Zag, while for the Gibbs sampler the expected number of steps is $d^2$. 
 \end{itemize} 
 Table~\ref{tab: 2} in Section~\ref{sec:runtime analysis} summarizes the scaling results derived in the two cases considered above. 
 
 \clearpage
 
\section{Details of Section~\ref{sec: examples}}\label{app: details of examples}
\subsection{Logistic regression}\label{app: logistic regression}
Similar computations for the bounds of the Poisson rates of the Zig-Zag sampler applied to logistic regressions can be found in the supplementary material of \textcite{bierkens2019}.
Given a posterior density of the form of Equation~\eqref{eq: target measure} with 
\[
\Psi(x) = \sum_{j=1}^N \left(\log \left(1 + e^{ \langle  A_{[j,:]}, x\rangle}\right) - y_j\langle A_{[j,:]}, x\rangle \right) + \frac{1}{2\sigma^{2}} \|x\|^2 
\]
we use the Sticky Zig-Zag subsampler presented in Section~\ref{sec: Subsampling}. To that end, define $U(x) = \Psi(x) - \frac{1}{2\sigma^{2}} \|x\|^2$. 
We decompose the partial derivatives of $U$ as follow:
\[
\partial_{x_i}U(x) = \sum_{j\in \Gamma_{i}} S(x, i, j)
\]
with sets $\Gamma_i = \{ j \in \{1,2,\dots, N\} \colon A_{j,i} \ne 0\}$ and
\[
S(x,i,j) = \left(\frac{A_{[j,i]} e^{ \langle  A_{[j,:]},x \rangle}}{ 1 +  e^{ \langle  A_{[j,:]}, x\rangle}} - y_j A_{[j,i]}\right).
\]
Then, for all $i = 1,2,\dots, p$ and any $x' \in \RR^p$, if J $\sim \text{Unif}(\Gamma_k)$, the estimator $[|\Gamma_i| (S(x,i,J) - S(x^{'},i,J)] + \partial_{x_i} U(x^{*}) + \sigma^{-2}x_i$ is unbiased for $ \partial_{x_i}\Psi(x)$.
 Notice that the partial derivative of $S(x,k, j)$ is bounded:
\[
\partial_{x_i}(S(x,k,j)) = \frac{A_{[j,k]} A_{[j,i]} e^{ \langle  A_{[j,:]},x \rangle}}{\left( 1 + e^{ \langle  A_{[j,:]},x \rangle} \right)^2} \le \frac14 A_{[j,k]}A_{[j,i]},
\]
which means that for $i = 1,2,\dots,d$
\begin{equation*}
\label{eq: lipschitz condition}
|S(x,i,j) - S(x',i,j)| \le C_i \|x - x'\|_p, \quad p \ge 1, \, j \in \Gamma_i, \, x, x' \in \RR^d,
\end{equation*}
with 
\[
C_k = \frac{1}{4}\max_{j = 1,..,N}|A_{[j,k]}|\,\|A_{j,:} \|_2. 
\]

Then given an initial position $(x,v) \in E$, tuning parameter $x'$ and for any $t\ge 0$, write $(x(t), v(t)) = \varphi(t, x, v)$ with $i \in \alpha(x,v)$ :
\begin{align*}
&  \tilde \lambda_i(x(t), v(t)) = \left(v_i\left(\partial_{x_i}U(x') + \sigma^{-2}x_i(t) + |\Gamma_i| (S(x(t), i, j) - S(x', i, j)) \right)\right)^+ \\
&\quad \le  (v_i(\partial_{x_i}U(x') + \sigma^{-2}(x_i + v_i t)))^+ +  |v_i| |\Gamma_i| \left(|S(x(t), i, j) -  S(x, i, j)| +| S(x, i, j) -  S(x', i ,j)|\right)\\
&\quad \le (v_i(\partial_{x_i}U(x') + \sigma^{-2}(x_i + v_i t))^+ + |v_i| |\Gamma_i| C_i\left(t \|v\|_p + \| x - x'\|_p\right).
\end{align*}

Thus we set
\[
\lambda_i(t, x, v ) = v_i (a_i(x, v) + b_i(x, v)t)
\]
where $a_i(x, v) = (v_i(\partial_i U(x') + \sigma^{-2}x_i))^+ + C_i |\Gamma_i| |v_i| \|x - x' \|_p $ and $b_i(x, v) = |v_i|C_i |\Gamma_i| \| v \|_p + v_i^2 \sigma^{-2}$.
We choose $x'$ to be the posterior mode of $\exp(-\Psi)$, which in this case is unique and easily found with the Newton's method since the function $\exp(-\Psi)$ is convex. 
Given an initial position $(x,v)$, suppose the particle $j \ne i$ gets frozen at time $\tau\ge 0$. Then for $t \ge \tau$ we have that $\|\int_0^t v(t) \dd t\|_p  = \tau\|v\|_p + (t - \tau)\|v'\|_p   \le t\|v\|_p$, with $v' = v[j \colon 0]$. This implies that the Poisson times drawn before the $j$th coordinate gets stuck are still valid upper bounds after time $\tau$. The same argument  follows easily for $n\ge1$ coordinates getting stuck at 0.
\subsection{Spatially structured sparsity}\label{app: spatially structured} 
For this application, we use the thinning scheme as presented in Appendix~\ref{app: computing Poisson times}. 
The bounding rates are of the form
\begin{equation}\label{eq: bounding rate heart}
\bar\lambda_i(t, x(t_0), v(t_0)) = (c+v_i(t_0) \partial_{x_i} \Psi(x(t_0))^+
\end{equation}
for $t \in [0, \Delta]$ with $\Delta =  1/c$.
To see this, define the Lipschitz growth bound $L_{x, v, \Delta}$ as
\[P(\sup_{0 < t<\Delta} \frac1t|V_i(t)\partial_{x_i} \Psi(X(t))| \le L_{x,\Delta} \mid X(0) = 0, V(0) = v) = 1, \quad i = 1,2,\dots,d,\]
which gives an explicit expression for $c$ in Equation~\eqref{eq: bounding rate heart} as
\[
c - L_\Delta\Delta = 0 \, \Rightarrow \, \Delta = 1/c, 
\]
such that the inequality \eqref{eq: lambda bound} holds.
With $L_\Delta = \sup_{x} L_{x, v, \Delta}$, in this application we have that
\[
L_\Delta = \sup_{v, t}  |\partial_t \partial_{x_i} \Psi(x + tv)| = c_2 + 8c_1 + 1/\sigma^2
\]
with $c_1,c_2$ defined in Section~\ref{sec: example>Spatially structured sparsity}. With this given choice, in the simulations of Section~\ref{sec: example>Spatially structured sparsity}, the ratio between the accepted reflection times and the proposed reflection times was 0.357. Here we used the local implementation of the Sticky Zig-Zag given by Appendix~\ref{app: local implementation} (with sets $\overline{A_i} = i$ for all $i$) in conjunction with the sparse algorithm as in Remark~\ref{rmk: sparse implementation}.
\subsection{Sparse precision matrix}\label{app:sparse_precision_matrix}

By write  $\Psi(x)\bigotimes_{i=1}^p \bigotimes_{j=1}^{i} (\dd x_{i,j} + \frac{1}{\kappa} \delta_0(\dd x_{i,j})\ind_{(i \ne j)})$ and we have that
\begin{equation}\label{eq: psi precision matrix}
\partial_{x_{i,j}} \Psi(x) = 
(Y Y')_{(i,:)} X_{(:, j)} + \gamma_{i,j} (x_{i,j} - c_{i,j})
 -  \ind_{(i = j)} \left(\frac{N}{x_{i,j}}\right).
\end{equation}

Note that, for any initial position and velocity $(x, v)$, the reflection times of the Sticky Zig-Zag with rates $\lambda_{i,j}(\phi(t, x, v)) = (v_i \partial_{x_{i,j}} \Psi(x + vt) )^+$ can be computed exactly for the off-diagonal elements and via a thinning scheme for the diagonal elements where
\[
\lambda_{i,i}(\phi(t, x, v) \le \overline{\lambda}_{i,i}(t, x,v) + \overline{\overline{\lambda}}_{i,i}(t, x,v),\quad t>0, \forall i.
\]
Here $\overline{\lambda}_{i,i}(t,x,v) = (v_{i,i}(YY'_{i,:} (X_{:,i} + vt) +\gamma_{i,i} (x_{i,i}+vt - c_{i,i})))^+$ and $\overline{\overline{\lambda}}_{i,i}(t, x,v) = \left(-v_{i,i}\frac{N}{x_{i,i} + v_{i,i} t}\right)$ and a Poisson time form the bounding rate is simulated as $\min(\tau_1,\tau_2)$ where $\tau_1 \sim \text{Poiss}(s \to \overline{\lambda}_{i,i}(s, x, v))$ and 
$\tau_2 \sim \text{Poiss}(s \to \overline{\overline{\lambda}}_{i,i}(s, x, v))$.
\end{document}